\documentclass[twoside,11pt]{article}

\usepackage{jmlr2e}

\usepackage[mathscr]{eucal}
\usepackage[cmex10]{amsmath}
\usepackage{epsfig,epsf,psfrag,amssymb,amsfonts,latexsym,url,amsmath,graphicx,bm,  xcolor,url}
\usepackage[caption=false]{subfig}
\usepackage{fixltx2e}
\usepackage{verbatim}
\usepackage{bm}
\usepackage{latexsym, amssymb, amsmath}

\allowdisplaybreaks[4]

\def\beq{\vskip .1cm \begin{equation}}
\def\eeq{\end{equation} \vskip .1cm \noindent}
\def\beqn{\vskip .1cm \begin{eqnarray}}
\def\eeqn{\end{eqnarray} \vskip .1cm \noindent}
\def\beqnn{\vskip .01cm \begin{eqnarray*}}
\def\eeqnn{\end{eqnarray*} \vskip .01cm \noindent}

\def\nn{\nonumber}

\newcommand{\nbd}{\mathrm{nbd}}

\newcommand{\CLThres}{\mathsf{CLThres}}
\newcommand{\calX}{\mathcal{X}}
\newcommand{\calW}{\mathcal{W}}
\newcommand{\calF}{\mathcal{F}}
\newcommand{\calY}{\mathcal{Y}}
\newcommand{\bX}{\mathbf{X}}

\newcommand{\bD}{\mathbf{D}}
\newcommand{\bW}{\mathbf{W}}
\newcommand{\bv}{\mathbf{v}}
\newcommand{\bPi}{\bm{\Pi}}

\newcommand{\bz}{\mathbf{z}}

\newcommand{\bN}{\mathbb{N}}

\newcommand{\bbT}{\mathsf{T}_n}

\newcommand{\eig}{\mathrm{eig}}
\newcommand{\MAP}{\mathrm{MAP}}
\newcommand{\vect}{\mathrm{vec}}

\newcommand{\hQ}{\widehat{Q}}

\newcommand{\hT}{\widehat{T}}

\newcommand{\tilP}{\widetilde{P}}
\newcommand{\tilM}{\widetilde{M}}
\newcommand{\tilcalC}{\widetilde{\calC}}

\newcommand{\diag}{\mathrm{diag}}

\newcommand{\calC}{\mathcal{C}}
\newcommand{\calR}{\mathcal{R}}
\newcommand{\calS}{\mathcal{S}}

\newcommand{\calE}{\mathcal{E}}

\newcommand{\calB}{\mathcal{B}}
\newcommand{\calT}{\mathcal{T}}

\newcommand{\calA}{\mathcal{A}}
\newcommand{\calD}{\mathcal{D}}
\newcommand{\bx}{\mathbf{x}}

\newcommand{\bR}{\mathbb{R}}
\newcommand{\calP}{\mathcal{P}}
\newcommand{\calU}{\mathcal{U}}
\newcommand{\bP}{\mathbb{P}}

\newcommand{\bH}{\mathbf{H}}

\newcommand{\bA}{\mathbf{A}}

\newcommand{\bC}{\mathbf{C}}

\newcommand{\argmin}{\operatornamewithlimits{argmin}}
\newcommand{\argmax}{\operatornamewithlimits{argmax}}
\newcommand{\hP}{\widehat{P}}
\newcommand{\hE}{\widehat{E}}
\newcommand{\he}{\widehat{e}}
\newcommand{\hk}{\widehat{k}}
\newcommand{\veps}{\varepsilon}

\newcommand{\poly}{\mathrm{poly}}

\jmlrheading{1}{2010}{1-48}{4/00}{10/00}{Vincent Tan, Animashree Anandkumar and Alan Willsky}

\ShortHeadings{Learning High-Dimensional Markov Forest  Distributions}{Tan, Anandkumar and Willsky}
\firstpageno{1}

\begin{document}

\title{Learning High-Dimensional Markov Forest  Distributions: Analysis of Error Rates}

\author{\name Vincent Y. F. Tan$^{*\dagger}$ \email    vtan@wisc.edu, vtan@mit.edu \\
\addr $^*$Department of Electrical and Computer Engineering,\\
 University of Wisconsin-Madison,\\
 Madison, WI 53706 \\\\
		   \name Animashree Anandkumar$^{\ddagger}$ \email a.anandkumar@uci.edu \\
		   \addr $^{\ddagger}$Center for Pervasive Communications and Computing,\\
Electrical Engineering and Computer Science,\\
University of California, Irvine,\\
Irvine, CA 92697\\\\
       \name Alan S. Willsky$^{\dagger}$ \email willsky@mit.edu \\
       \addr $^{\dagger}$Stochastic Systems Group,\\Laboratory for Information and Decision Systems,\\
       Massachusetts Institute of Technology,\\
       Cambridge, MA 02139}

\editor{Marina Meil\u{a}}

\maketitle

\begin{abstract}
The problem of learning forest-structured discrete graphical models from i.i.d.\ samples is considered. An  algorithm based on pruning of the Chow-Liu tree through adaptive thresholding is proposed.   It is shown that this algorithm is both  structurally consistent and risk consistent and  the error probability of structure learning decays faster than any polynomial in the number of samples   under fixed model size.  For the  high-dimensional scenario where the size of the  model $d$ and the number of edges $k$ scale with the number of samples $n$,  sufficient conditions on $(n,d,k)$ are given for  the algorithm to satisfy structural and risk consistencies. In addition, the extremal structures for learning are identified; we prove that the independent (resp.\ tree) model is the hardest (resp.\ easiest) to learn using the proposed algorithm in terms of error rates for structure learning.
\end{abstract}


\section{Introduction} \label{sec:intro}


Graphical models (also known as Markov random fields) have a wide range of applications in diverse fields such  as signal processing, coding theory and bioinformatics. See \cite{Lau96}, \cite{Wai03} and references therein for examples. Inferring the structure and parameters of graphical models from samples is a starting point in all these applications. The structure of the model provides a quantitative interpretation of relationships amongst the given collection of random variables by specifying a set of conditional independence relationships. The parameters of the model quantify the strength of these interactions among the variables. 

The challenge in learning graphical models is often compounded by the fact that typically only a small number of samples are  available relative to the size of the model (dimension of data). This is referred to as the high-dimensional learning regime, which differs from classical statistics where a large number of samples of fixed dimensionality are available. As a concrete example, in order to analyze the effect of environmental and genetic factors on childhood asthma, clinician scientists in Manchester, UK  have been conducting a longitudinal birth-cohort study since 1997 \citep{Custovic, Simpson}. The number of variables collected is of the order of $d\approx  10^6$ (dominated by the genetic data) but the number of children in the study is small ($n\approx 10^3$). The paucity of subjects in the study is  due in part to the prohibitive cost of collecting high-quality clinical data from willing participants. 

In order to learn high-dimensional graphical models,  it is imperative to strike the right balance between data fidelity and overfitting. To ameliorate the effect of overfitting, the samples are often fitted to a {\em sparse graphical model}  \citep{Wai03}, with a small number of edges.  One popular and tractable class of sparse graphical models is the set of tree\footnote{A {\em tree} is a   {\em connected}, acyclic graph. We use the term {\em proper forest} to denote the set of {\em disconnected}, acyclic graphs. } models. When restricted to trees, the Chow-Liu algorithm \citep{CL68,Cho73} provides an efficient implementation of the maximum-likelihood (ML) procedure to learn the structure from independent  samples.   However, in the high-dimensional regime, even a tree may overfit the data \citep{Gupta}. In this paper, we consider learning high-dimensional, forest-structured (discrete) graphical models from a given  set of samples. 



For learning the  forest structure,  the ML (Chow-Liu) algorithm does not produce a consistent estimate  since ML favors richer model classes and hence, outputs a tree in general. We propose a consistent algorithm  called $\CLThres$, which has a thresholding mechanism to prune ``weak'' edges from the  Chow-Liu tree.   We provide tight bounds on the {\em overestimation} and {\em underestimation} errors, that is, the error probability that the output of the algorithm has  more or fewer edges than the true model. 

\subsection{Main Contributions}
This paper contains three main contributions. Firstly, we propose an algorithm named $\CLThres$   and prove that it is structurally consistent when the true distribution is  forest-structured. Secondly, we prove that $\CLThres$ is risk consistent, meaning that the risk under the estimated model converges to the risk of the {\em forest projection}\footnote{The forest projection is the forest-structured graphical model that is closest in the KL-divergence sense to the true distribution. We define this distribution formally in \eqref{eqn:forest_proj}.} of the underlying distribution, which may not be a forest. We also provide precise convergence  rates for  structural and risk consistencies. Thirdly, we provide conditions for the consistency of $\CLThres$ in the high-dimensional setting. 

We first prove that $\CLThres$ is structurally consistent, i.e., as the number of samples grows for a fixed model size, the probability of learning the incorrect structure (set of edges), decays to zero for a fixed model size. We show that the error rate is in fact, dominated by the rate of decay of the overestimation error probability.\footnote{The overestimation error probability is the probability that the number of edges learned exceeds the true number of edges. The underestimation error is defined analogously.} We use an information-theoretic technique  known as the {\em method of types} \citep[Ch.\ 11]{Cov06} as well as a recently-developed technique known as Euclidean information theory \citep{Bor08}. We provide an upper bound on the error probability by using convex duality to find a surprising connection between the overestimation error rate and a semidefinite program \citep{Van96} and  show that the overestimation error in structure learning decays faster than any polynomial in $n$ for a fixed data dimension $d$. 

We then consider the high-dimensional scenario and provide sufficient conditions on  the growth of $(n,d)$ (and also the true number of edges $k$) to ensure that $\CLThres$ is structurally consistent. We prove that even if $d$ grows faster than any polynomial in $n$ (and in fact close to exponential in $n$), structure estimation remains consistent. As a corollary from our analyses, we also show that for $\CLThres$, independent models (resp.\  tree models) are the ``hardest'' (resp.\ ``easiest'') to learn in the sense that the asymptotic error rate is the highest (resp.\ lowest), over all models with the same scaling of $(n,d)$. Thus, the empty graph and connected trees are the  extremal forest structures  for learning.  We also prove that  $\CLThres$ is risk consistent, i.e., the risk of the  estimated forest  distribution converges to the risk of the forest projection of the true model at a rate of $O_p(d\log d/n^{1-\gamma})$ for any $\gamma>0$. We compare and contrast this rate to existing results such as those \cite{Gupta}. Note that for this result, the true probability model does not need to be a forest-structured distribution. Finally, we use  $\CLThres$ to learn forest-structured distributions given synthetic and real-world datasets and show that in the finite-sample case, there exists an inevitable trade-off between the underestimation and overestimation errors. 


\subsection{Related Work}
There are many papers that discuss the learning  of graphical models from data. See \cite{Dud04}, \cite{Lee06}, \cite{Abb06}, \cite{Wai06}, \cite{Mei06}, \cite{Joh07},  and references therein. Most of these methods pose the learning problem as a parameterized  convex optimization problem, typically with a regularization term to enforce sparsity in the learned graph.  Consistency guarantees in terms of $n$ and $d$ (and possibly the maximum degree) are provided. Information-theoretic limits for learning graphical models have also been derived in \cite{San08}. In \cite{Zuk06}, bounds on the error rate for learning the structure of Bayesian networks using the Bayesian Information Criterion (BIC) were provided.  \cite{Bach03} learned tree-structured models for solving the independent component analysis (ICA) problem.  A PAC analysis for learning  thin junction trees was given in \cite{Chechetka07}. \cite{Mei00} discussed the learning of graphical models from a different perspective; namely that of learning mixtures of trees via an expectation-maximization procedure. 

By using the theory of large-deviations \citep{Dembo}, we derived and analyzed the error exponent for learning trees for discrete \citep{Tan&etal:09ITsub} and Gaussian \citep{Tan&etal:09SPsub} graphical models. The error exponent is a quantitative measure of performance of the learning algorithm since a larger exponent implies a faster decay of the error probability. However, the analysis does not readily extend to learning forest models. In addition, we posed the structure learning problem for trees as a composite hypothesis testing problem \citep{Tan10:ISIT}  and derived a closed-form expression for the Chernoff-Stein exponent in terms of the mutual information on the bottleneck edge.  In a paper that is   closely related to ours, \cite{Gupta} derived consistency (and sparsistency) guarantees for learning tree and forest models. The pairwise joint distributions are modeled using kernel density estimates, where the kernels are H\"{o}lder continuous. This differs from our approach since we assume that each variable can only take finitely many values, leading to stronger results on error rates for structure learning via the method of types, a powerful proof technique in information theory and statistics. We compare our convergence rates to these related works in Section~\ref{sec:param_consistency}. Furthermore, the algorithm suggested in both papers uses a subset (usually half) of the dataset  to learn the full tree model and then uses the remaining subset  to prune the model based on the log-likelihood on the held-out set. We suggest a more direct and consistent method based on thresholding, which uses the {\em entire} dataset to learn and prune the model without recourse to validation on a held-out dataset. It is well known that validation is both computationally expensive \citep[pp.\ 33]{Bishop} and a potential waste of valuable data which may otherwise be employed to learn a better model.  In \cite{Gupta}, the problem of estimating forests with restricted component sizes was considered and was proven to be NP-hard. We do not restrict the component size in this paper but instead attempt to learn the model with the minimum number of edges which best fits the data.

Our work is also related to and inspired by the  vast body of  literature in information theory and statistics  on Markov order estimation. In these works, the authors use various regularization and model selection schemes to find the optimal order of a Markov chain \citep{Merhav89,Finesso96, CsiszarBIC2000}, hidden Markov model \citep{Gassiat} or exponential family \citep{Merhav89b}. We build on some of these ideas and proof techniques to identify the correct set of edges (and in particular the number of edges) in the forest model and also to provide strong theoretical guarantees of the rate of convergence of the estimated forest-structured distribution to the true one.

\subsection{Organization of Paper}
This paper is organized as follows: We define the mathematical notation and formally state the problem in Section~\ref{sec:formulation}. In Section~\ref{sec:algorithm}, we describe the algorithm in full detail, highlighting its most salient aspect -- the thresholding step. We state our main results on error rates for structure learning in Section~\ref{sec:structural_consistency} for a fixed forest-structured distribution. We extend these results to the high-dimensional case when $(n,d,k)$ scale in Section~\ref{sec:structural_high}. Extensions to rates of convergence of the estimated distribution, i.e., the order of risk consistency, are discussed briefly in Section~\ref{sec:param_consistency}. Numerical simulations on synthetic and real data are presented in Section~\ref{sec:numerical}. Finally, we conclude the discussion in Section~\ref{sec:conclusion}. The proofs of the majority of  the results are provided in the appendices.

\section{Preliminaries and Problem Formulation} \label{sec:formulation}
Let $G=(V,E)$ be an undirected graph with vertex (or node) set $V:=\{1,\ldots, d\}$ and edge set $E\subset \binom{V}{2}$ and let $\nbd(i) :=\{j\in V: (i,j)\in E\}$ be the set of  neighbors of vertex $i$.  Let the set of labeled {\em  trees} (connected, acyclic graphs) with $d$ nodes be $\calT^d$ and let the set of {\em forests} (acyclic graphs) with $k$ edges and $d$ nodes be $\calT_k^d$ for $0\le k\le d-1$. The set of forests includes all the trees. We  reserve the term {\em proper forests} for the set of disconnected acylic graphs $\cup_{k=0}^{d-2}  \calT_k^d$. We also use the notation $\calF^d:=\cup_{k=0}^{d-1}\calT_k^d$ to denote the set of  labeled forests with $d$ nodes. 

A {\em graphical model} \citep{Lau96} is a family of multivariate probability distributions (probability mass functions) in which each distribution factorizes according to a given undirected graph and where each variable is associated to a node in the graph. Let $\calX = \{1,\ldots, r\}$ (where $2\le r<\infty$) be a finite set and $\calX^d$ the $d$-fold Cartesian product of the set $\calX$. As usual, let $\calP(\calX^d)$ denote the probability simplex over the  alphabet $\calX^d$. We say that the random vector $\bX=(X_1,\ldots,X_d)$ with distribution  $Q\in\calP(\calX^d)$ is {\em Markov on} the graph $G=(V,E)$ if 
\begin{equation}
Q(x_i|x_{\nbd(i)})= Q(x_i|x_{V\setminus i}),\qquad\forall\, i\in V, \label{eqn:localMarkov}
\end{equation} 
where $x_{V\setminus i}$ is the collection of variables excluding variable $i$. Eq.~\eqref{eqn:localMarkov} is known as the {\em local Markov property}~\citep{Lau96}. In this paper, we always assume that  graphs are {\em minimal representations} for the corresponding graphical model, i.e., if $Q$ is Markov on $G$, then $G$ has the smallest number of edges for the conditional independence relations in~\eqref{eqn:localMarkov} to hold.  We say the distribution $Q$ is a {\em forest-structured distribution} if it is Markov on a forest. We also use the notation $\calD(\calT_k^d) \subset \calP(\calX^d)$ to denote the set of $d$-variate distributions Markov on a forest with $k$ edges.  Similarly, $\calD(\calF^d)$ is the set of forest-structured distributions. 

Let $P\in \calD(\calT_k^d)$ be a discrete forest-structured distribution Markov on $T_P = (V,E_P)\in\calT_k^d$ (for some $k=0,\ldots, d-1$). It is known that the joint distribution $P$ factorizes as follows~\citep{Lau96}:
\begin{equation}
P(\bx) =\prod_{i\in V} P_{i}(x_i)  \prod_{(i,j)\in E_P}\frac{P_{i,j}(x_i,x_j) }{P_{i}(x_i) P_{j}(x_j) }, \label{eqn:tree_decomp} 
\end{equation}
where $\{P_i\}_{i\in V}$ and $\{P_{i,j}\}_{(i,j)\in E_P}$ are the node and pairwise marginals which are assumed to be positive everywhere.    


The mutual information (MI)  of two random variables $X_i$ and $X_j$ with joint distribution $P_{i,j}$ is the function $I(\cdot):\calP(\calX^2)\to[0,\infty)$ defined as 
\begin{equation}
I(P_{i,j}) := \sum_{x_i,x_j\in\calX} P_{i,j}(x_i,x_j)  \log\frac{P_{i,j}(x_i,x_j)}{P_i(x_i) P_j(x_j)}. \label{eqn:mi_def}
\end{equation}
This notation for mutual information differs from the usual $I(X_i;X_j)$ used in \cite{Cov06}; we emphasize the dependence of $I $ on the joint distribution $P_{i,j}$. The {\em minimum mutual information} in the forest, denoted as $I_{\min} :=\min_{(i,j)\in E_P}  I(P_{i,j})$  will turn out to be a fundamental quantity in the subsequent analysis. Note from our minimality assumption that $I_{\min}>0$ since all edges in the forest have positive mutual information (none of the edges are degenerate).  When we consider the scenario where $d$ grows with $n$ in Section~\ref{sec:structural_high}, we assume that $I_{\min}$ is {\em uniformly} bounded away from zero. 

\subsection{Problem Statement}
We now state the basic problem formally. We are given  a set of i.i.d.\ samples, denoted as $\bx^n:=\{\bx_1,\ldots, \bx_n\}$. Each sample $\bx_l=(x_{l,1},\ldots, x_{l,d})\in\calX^d$ is drawn independently from $P\in\calD(\calT_k^d)$ a forest-structured distribution. From these samples, and the prior knowledge that the undirected graph is acyclic (but not necessarily connected), estimate the true set of edges $E_P$ as well as the true distribution $P$ consistently. 

\section{The Forest Learning Algorithm: $\CLThres$}\label{sec:algorithm}
We now describe our algorithm for estimating the edge set  $E_P$ and the distribution $P$. This algorithm is a modification of the celebrated Chow-Liu algorithm for maximum-likelihood (ML) learning of  tree-structured distributions \citep{CL68}. We call our algorithm  $\CLThres$ which stands for {\em  Chow-Liu with Thresholding}.

The inputs to the algorithm are the set of samples $\bx^n$ and a  {\em regularization} sequence $\{\veps_n\}_{n\in\bN}$ (to be specified precisely later) that typically decays to zero, i.e., $\lim_{n\to\infty}\veps_n= 0$. The outputs are the estimated edge set, denoted $\hE_{\hk_n}$, and the estimated distribution, denoted $P^*$. 
\begin{enumerate}
\item  \label{item:emp} Given $\bx^n$, calculate the set of  {\em pairwise empirical distributions}\footnote{In this paper, the terms {\em empirical distribution} and {\em type} are used interchangeably.} (or {\em pairwise types})  $\{\hP_{i,j}\}_{i,j\in V}$.
This is just a normalized version of the counts of each observed symbol in $\calX^2$ and serves as a set of  sufficient statistics for the estimation problem. The dependence of $\hP_{i,j}$ on the samples $\bx^n$ is suppressed.
\item \label{item:mi}  Form the set of {\em empirical mutual information} quantities:
\begin{eqnarray}
I(\hP_{i,j}) := \sum_{(x_i,x_j)\in\calX^2} \hP_{i,j}(x_i,x_j) \log\frac{\hP_{i,j}(x_i,x_j)}{\hP_{i}(x_i)\hP_{j}(x_j)}, \nn
\end{eqnarray}
for $1\le i,j\le d$. This is a consistent estimator of the true mutual information in~\eqref{eqn:mi_def}.
\item \label{item:mwst} Run a max-weight spanning tree (MWST) algorithm~\citep{Prim, Kruskal} to obtain an estimate of the edge set:
\begin{equation}
\hE_{d-1} := \argmax_{E: T=(V,E) \in\calT^d} \,\,\sum_{(i,j) \in E}  I(\hP_{i,j}). \nn
\end{equation}
Let the  estimated  edge set be $\hE_{d-1} :=\{\he_1, \ldots, \he_{d-1}\}$ where the  edges $\he_i$ are sorted according to decreasing  empirical mutual information values. We index the edge set by $d-1$ to emphasize that it has $d-1$ edges and hence is connected. We denote the sorted empirical mutual information quantities as $I(\hP_{\he_1})  \ge \ldots\ge I(\hP_{\he_{d-1}})$. These first three steps constitute the Chow-Liu algorithm \citep{CL68}. 
\item \label{item:thres} Estimate the true number of edges using the {\em thresholding estimator}:
\begin{eqnarray}
\hk_n   :=  \argmin_{1\le j\le d-1}   \left\{I(\hP_{\he_j}) :  I(\hP_{\he_j}) \ge  \veps_n , I(\hP_{\he_{j+1}})  \le  \veps_n  \right\}.    \label{eqn:kn}
\end{eqnarray}
If there exists an empirical mutual information  $I(\hP_{\he_j})$ such that $I(\hP_{\he_j})=\veps_n$, break the tie arbitrarily.\footnote{Here were allow a bit of imprecision by noting that the non-strict inequalities in~\eqref{eqn:kn} simplify the subsequent analyses because the constraint sets that appear in optimization problems will be closed, hence compact, insuring the existence of optimizers.}
\item \label{item:prune}  Prune the tree by retaining only the top $\hk_n$ edges, i.e., define the {\em  estimated edge set} of the forest to be
\begin{equation}
\hE_{\hk_n} :=\{\he_1, \ldots, \he_{\hk_n}\}, \nn
\end{equation}
where $\{\he_i:1\le i\le d-1\}$ is the ordered edge set defined in Step~\ref{item:mwst}. Define the estimated forest to be $\hT_{\hk_n} := (V, \hE_{\hk_n})$.
\item \label{item:Pstar} Finally, define the estimated distribution $P^*$ to be the {\em reverse I-projection}~\citep{Csis:InfoProj}   of the joint type $\hP$ onto $\hT_{\hk_n}$, i.e., 
\begin{equation}
P^*(\bx):= \argmin_{Q\in\calD(\hT_{\hk_n})}  \,\,D(\hP\,||\, Q). \nn 
\end{equation}
It can easily be shown that the projection  can be expressed in terms of the marginal and pairwise joint types:
\begin{equation}
P^*(\bx)=\prod_{i\in V}\hP_i(x_i)\prod_{(i,j)\in\hE_{\hk_n} } \frac{\hP_{i,j}(x_i,x_j)}{\hP_{i}(x_i)\hP_{j}(x_j)}. \nn 
\end{equation}
\end{enumerate}
Intuitively, $\CLThres$ first constructs a connected tree $(V, \hE_{d-1})$ via Chow-Liu (in Steps \ref{item:emp} -- \ref{item:mwst}) before pruning the weak edges (with small mutual information) to obtain the final  structure $\hE_{\hk_n}$. The estimated distribution $P^*$ is simply the ML estimate of the parameters subject to  the constraint that $P^*$ is Markov on the learned tree  $\hT_{\hk_n}$.

Note that if Step~\ref{item:thres} is omitted and $\hk_n$ is defined to be  $d-1$, then  $\CLThres$ simply reduces to the Chow-Liu ML algorithm. Of course Chow-Liu, which outputs a  tree, is guaranteed to fail (not be structurally consistent) if the number of edges in the true model $k<d-1$, which is the problem of interest in this paper.  Thus, Step~\ref{item:thres}, a model selection step, is essential in estimating the true number of edges $k$. This step is a generalization of the  test for independence of discrete memoryless sources discussed in \cite{Merhav89b}. In our work, we exploit the fact that the empirical mutual information  $I(\hP_{\he_j})$   corresponding to a pair of independent variables $\he_{j}$ will be very small when $n$ is large, thus a thresholding procedure using the (appropriately chosen)  regularization sequence $\{\veps_n\}$ will remove these edges.  In fact, the subsequent analysis allows us to conclude that Step \ref{item:thres}, in a formal sense, {\em dominates} the error probability in structure learning. $\CLThres$  is also efficient as shown by the following result.
\begin{proposition}[Complexity of $\CLThres$] \label{prop:complexity}
$\CLThres$ runs in time $O((n+\log d) d^2)$. 
\end{proposition}
\begin{proof}
The computation of the sufficient statistics in Steps \ref{item:emp} and \ref{item:mi} requires $O(nd^2)$ operations. The MWST algorithm in Step \ref{item:mwst} requires at most $O(d^2 \log d)$ operations \citep{Prim}. Steps \ref{item:thres} and \ref{item:prune} simply require the sorting of the empirical mutual information quantities on the learned tree which only requires $O(\log d)$ computations. 
\end{proof}

\section{Structural Consistency For Fixed Model Size} \label{sec:structural_consistency}
In this section, we keep $d$ and $k$ fixed and  consider a  probability model $P$, which is assumed to be Markov on a forest in $\calT_k^d$. This is to gain better insight into the problem before we analyze the high-dimensional scenario in Section~\ref{sec:structural_high} where $d$ and $k$ scale\footnote{In that case $P$ must also scale, i.e., we learn a {\em family} of models as $d$ and $k$ scale.} with the sample size $n$. More precisely, we are interested in quantifying the rate at which the probability of the  error event of structure learning
\begin{equation}
\calA_n:=\left\{\bx^n\in(\calX^d)^n : \hE_{\hk_n}(\bx^n)   \ne E_P \right\}  \label{eqn:An}
\end{equation}
decays to zero as $n$ tends to infinity. Recall that $\hE_{\hk_n}$, with cardinality $\hk_n$, is the learned edge set  by using $\CLThres$.  As usual, $P^n$ is the $n$-fold product probability measure corresponding to the forest-structured distribution $P$. 

Before stating the main result of this section in Theorem~\ref{thm:forest}, we first state an auxiliary result that essentially says that if one is provided with oracle knowledge of $I_{\min}$, the minimum mutual information in the forest, then the problem is greatly simplified. 
\begin{proposition}[Error Rate with knowledge of $I_{\min}$] \label{prop:Imin}
Assume that $I_{\min}$ is known  in  $\CLThres$. Then by letting the regularization sequence be $\veps_n=I_{\min}/2$ for all $n$, we have
\begin{equation}
\lim_{n\to\infty}\frac{1}{n}\log P^n(\calA_n)< 0,  \label{eqn:expo_Imin}
\end{equation}
i.e., the error probability decays exponentially fast. 
\end{proposition}
The proof of this theorem and all other results in the sequel can be found in the appendices. 

Thus, the primary difficulty lies in estimating  $I_{\min}$ or equivalently, the number of edges $k$. Note that if $k$ is known, a simple modification to the Chow-Liu procedure by imposing the constraint that the final structure contains $k$ edges will also yield exponential decay as in~\eqref{eqn:expo_Imin}.   However, in the realistic case where both $I_{\min}$ and $k$ are  unknown, we show in the rest of this  section that we can design the regularization sequence $\veps_n$ in such a way that the rate of decay of $P^n(\calA_n)$ decays almost exponentially fast. 

\subsection{Error Rate for Forest Structure Learning}
We now  state one of the main results in this paper. We emphasize that the following result  is stated for  a fixed forest-structured distribution $P\in\calD(\calT_k^d)$ so $d$ and $k$ are also fixed natural numbers.

\begin{theorem}[Error Rate for   Structure Learning] \label{thm:forest}
Assume that the regularization sequence $\{\veps_n\}_{n\in\bN}$   satisfies the following two conditions:
\begin{equation}
\lim_{n\to\infty} \, \veps_n=0,\qquad \lim_{n\to\infty} \, \frac{n\veps_n}{\log n}=\infty. \label{eqn:veps} 
\end{equation}
Then,  if  the true model $T_P=(V,E_P)$ is a proper forest ($k<d-1$), there exists a constant $C_P \in (1,\infty)$ such that 
\begin{align}
-C_P &\le \liminf_{n\to\infty} \frac{1}{n\veps_n}\log P^n(\calA_n) \label{eqn:liminf_statement}\\
&\le \limsup_{n\to\infty} \frac{1}{n\veps_n}\log P^n(\calA_n)\le -1. \label{eqn:limsup_statement}
\end{align}
Finally, if the true model $T_P=(V,E_P)$   is a  tree ($k=d-1$), then 
\begin{equation}
\lim_{n\to\infty} \frac{1}{n }\log P^n(\calA_n) <0, \label{eqn:exponentially_fast}
\end{equation}
i.e.,  the error probability decays exponentially fast. 
\end{theorem}
\begin{figure}
\centering
\begin{picture}(160,150)
\multiput(30,5)(0,4){24}{\line(0,1){2}}
\linethickness{.6pt}
\put(-5, 5){\vector(1,0){160}}
\put(0,0){\vector(0,1){150}}
\put(157, 5){\mbox{$n$}}
\put(0,100){\line(1,0){150}}
\qbezier(7,150)(7,82)(150,40)
\qbezier(4,150)(3,17)(150,14)
\put(130, 105){\mbox{$I_{\min}$}}
\put(85, 67){\mbox{$\veps_n=\omega(\frac{\log n}{n})$}}
\put(111, 21){\mbox{$I(\hQ_{i,j}^n)\!\approx\!\frac{1}{n}$}}
\put(32, 7){\mbox{$N$}}
\end{picture}
\caption{Graphical interpretation of the condition on $\veps_n$. As $n\to\infty$, the regularization sequence $\veps_n$ will be smaller than $I_{\min}$ and larger than $I(\hQ_{i,j}^n)$ with high probability.}
\label{fig:rate_graphs}
\end{figure}
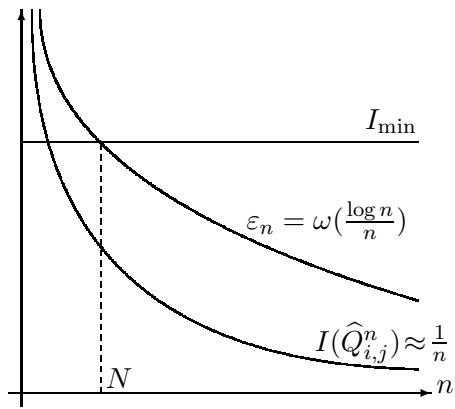

\subsection{Interpretation of Result}
From~\eqref{eqn:limsup_statement}, the rate of decay of the error probability for proper forests is subexponential but nonetheless can be made faster than any polynomial for an appropriate choice of $\veps_n$. The reason for the subexponential rate is because of our lack of knowledge of $I_{\min}$, the minimum mutual information in the true forest $T_P$. For trees, the rate\footnote{We use the asymptotic notation from information theory  $\doteq$ to denote equality  to first order in the exponent.  More precisely, for two positive sequences $\{a_n\}_{n\in\bN}$ and $\{b_n\}_{n\in\bN}$ we say that  $a_n\doteq b_n$ iff $\lim_{n\to\infty}n^{-1}\log (a_n/b_n)=0$.}  is exponential ($\doteq \exp(-nF)$ for some positive constant $F$). Learning proper forests is thus, strictly ``harder'' than learning trees. The condition on $\veps_n$ in \eqref{eqn:veps} is needed for the following intuitive reasons:
\begin{enumerate}
\item Firstly,  \eqref{eqn:veps} ensures that for all sufficiently large $n$, we have  $\veps_n < I_{\min}$. Thus, the true edges will be correctly identified by $\CLThres$ implying that with high probability, there will not be underestimation as $n\to\infty$. 
\item  Secondly,   for two independent  random variables $X_i$ and $X_j$ with  distribution $Q_{i,j}=Q_iQ_j$, the sequence\footnote{The notation $\sigma(Z)$ denotes the standard deviation of the random variable $Z$. The fact that the standard deviation of the empirical MI $\sigma(I(\hQ_{i,j}^n))$ decays as $1/n$ can be verified by Taylor expanding $I(\hQ_{i,j}^n)$ around $Q_{i,j} = Q_i Q_j$ and using the fact that the ML estimate converges at a rate of $n^{-1/2}$ \citep{Serfling}. } $\sigma(I(\hQ_{i,j}^n))=\Theta(1/n)$,  where $\hQ_{i,j}^n$ is the joint empirical distribution of $n$ i.i.d.\ samples drawn from $Q_{i,j}$. Since the regularization sequence $\veps_n = \omega( {\log n}/{n})$ has a slower rate of decay than  $\sigma (I(\hQ_{i,j}^n))$,  $\veps_n > I(\hQ_{i,j}^n )$  with high probability as $n\to\infty$. Thus, with high probability  there will not be overestimation as $n\to\infty$. 
\end{enumerate}
See Figure~\ref{fig:rate_graphs} for an illustration of this intuition. The formal proof follows from a method of types argument and we provide an outline in Section~\ref{sec:proofidea}.  A convenient choice of $\veps_n$ that satisfies $\eqref{eqn:veps}$ is 
\begin{equation}
\veps_n :=  n^{-\beta},\qquad \forall\,\beta\in (0,1). \label{eqn:vepsbeta}
\end{equation}

Note further that the upper bound in \eqref{eqn:limsup_statement} is also independent of $P$ since it is equal to $-1$ for all $P$. Thus,~\eqref{eqn:limsup_statement}  is  a {\em universal} result for all forest distributions $P\in \calD(\calF^d)$. The intuition for this universality is because in the large-$n$ regime, the typical way an error occurs is due to  overestimation. The overestimation error  results from testing whether pairs of random variables are independent and our asymptotic bound for the error probability of this test does not depend on the true distribution $P$. 

The lower bound $C_P$ in~\eqref{eqn:liminf_statement}, defined in the proof  in Appendix~\ref{prf:thm:forest}, means that we cannot hope to do much better using  $\CLThres$  if the original structure (edge set) is a proper forest. Together,~\eqref{eqn:liminf_statement} and~\eqref{eqn:limsup_statement} imply that the rate of decay of the error probability for structure learning is tight to within a constant factor in the exponent. We believe that the error rates given in Theorem~\ref{thm:forest} cannot, in general, be improved without knowledge of $I_{\min}$. We state a converse  (a necessary lower bound on sample complexity) in Theorem~\ref{thm:sample} by treating the unknown forest graph as a uniform random variable over all possible forests of fixed size.

\subsection{Proof Idea} \label{sec:proofidea}
The method of proof for Theorem~\ref{thm:forest} involves using the Gallager-Fano bounding technique \cite[pp.\ 24]{Fano} and the union bound to decompose  the overall error probability $P^n(\calA_n)$ into three  distinct terms: (i) the rate of decay of the error probability for learning the top $k$ edges (in terms of the   mutual information quantities) correctly -- known as the {\em Chow-Liu error}, (ii) the rate of decay of the {\em overestimation error} $\{\hk_n>k\}$ and  (iii) the rate of decay of the {\em underestimation error} $\{ \hk_n<k\}$. Each of these terms is upper bounded using a method of types~\cite[Ch.\ 11]{Cov06} argument. It turns out, as is the case with the literature on Markov order estimation (e.g., \cite{Finesso96}), that bounding the overestimation error poses the greatest challenge. Indeed, we show that the underestimation and Chow-Liu errors have exponential decay in $n$. However, the overestimation error has subexponential decay ($\approx \exp(-n\veps_n)$).

The main technique used to analyze the overestimation error relies on {\em Euclidean information theory}~\citep{Bor08} which states that if two distributions $\nu_0$ and $\nu_1$ (both supported on a common finite alphabet $\calY$) are close entry-wise, then various information-theoretic measures can be approximated locally by quantities related to Euclidean norms. For example, the KL-divergence $D(\nu_0 \, ||\, \nu_1)$ can be approximated by the square of a weighted Euclidean  norm:
\begin{eqnarray}
D(\nu_0\,||\, \nu_1) = \frac{1}{2} \sum_{a\in\calY} \frac{(\nu_0(a)  -  \nu_1(a))^2}{\nu_0(a)} +  o(\|\nu_0  -  \nu_1\|_{\infty}^2).   \label{eqn:dqp_app}
\end{eqnarray} 
Note that if $\nu_0\approx\nu_1$, then $D(\nu_0\,||\, \nu_1)$ is close to the sum in \eqref{eqn:dqp_app} and the  $o(\|\nu_0  -  \nu_1\|_{\infty}^2)$ term can be neglected. Using this approximation and Lagrangian duality \citep{Ber99}, we reduce a non-convex I-projection \citep{Csis:InfoProj} problem involving information-theoretic quantities (such as divergence) to a relatively simple {\em semidefinite program} \citep{Van96} which admits a closed-form solution.  Furthermore,  the approximation in~\eqref{eqn:dqp_app} becomes {\em exact} as $n\to \infty$ (i.e.,  $\veps_n\to 0$), which is the asymptotic regime of interest.  The full details of the proof can be found Appendix~\ref{prf:thm:forest}.

\subsection{Error Rate for Learning the Forest Projection}
In our discussion thus far, $P$ has been assumed to be Markov on a forest. In this subsection, we consider the situation when the underlying unknown distribution $P$ is not forest-structured but we wish to learn its best  forest approximation. To this end, we define the projection of $P$ onto the set of forests (or {\em forest projection}) to be 
\begin{equation}
\tilP :=  \argmin_{ Q\in \calD(\calF^d) }\,\, D(P \, ||\, Q). \label{eqn:forest_proj}
\end{equation}
If there are multiple optimizing distribution, choose a projection $\tilP$ that is minimal, i.e., its graph $T_{\tilP} = (V,E_{\tilP})$ has the {\em fewest number of edges} such that \eqref{eqn:forest_proj} holds. If we redefine the event $\calA_n$ in \eqref{eqn:An} to be $\widetilde{\calA}_n:= \{ \hE_{\hk_n}\ne E_{\tilP}\}$, we have the following analogue of Theorem~\ref{thm:forest}.

\begin{corollary}[Error Rate for Learning Forest Projection] \label{cor:forest_proj}
Let $P$ be an arbitrary distribution and the event $\widetilde{\calA}_n$ be defined as above. Then the conclusions in~\eqref{eqn:liminf_statement} -- \eqref{eqn:exponentially_fast} in Theorem~\ref{thm:forest} hold if the regularization sequence  $\{\veps_n\}_{n\in\bN}$   satisfies~\eqref{eqn:veps}. 
\end{corollary}

\section{High-Dimensional Structural Consistency} \label{sec:structural_high}
In the previous section, we considered learning a fixed forest-structured distribution $P$ (and hence fixed $d$ and $k$) and derived bounds on the error rate for structure learning.  However, for most problems of practical interest, the number of data samples is small compared to the data dimension $d$ (see the asthma example in the introduction). In this section, we prove sufficient conditions on the scaling of $(n,d,k)$ for structure learning to remain consistent. We will see that even if $d$ and $k$ are much larger than $n$, under some reasonable regularity conditions, structure learning remains consistent.

\subsection{Structure Scaling Law}
To pose the learning problem formally, we consider a {\em sequence} of structure learning problems indexed by the number of data points $n$. For the particular problem indexed by $n$, we have a dataset $\bx^n = (\bx_1,\ldots, \bx_n)$ of size $n$ where each sample $\bx_l\in\calX^d$ is drawn independently from an unknown $d$-variate forest-structured distribution $P^{(d)}\in \calD(\calT_k^d)$, which has $d$ nodes and $k$ edges and where $d$ and $k$ depend on $n$. This {\em high-dimensional} setup allows us to model and subsequently analyze how $d$ and $k$ can scale with $n$ while maintaining consistency. We will sometimes  make the dependence of $d$ and $k$ on $n$ explicit, i.e., $d=d_n$ and $k=k_n$.

In order to be able to learn the structure of the  models we assume that 
\begin{align}
\mbox{(A1)}\quad I_{\inf} := \inf_{d\in\bN}\,\min_{(i,j)\in E_{P^{(d)}}} \,  I(P^{(d)}_{i,j} )>0, \label{eqn:minMI} \\
\mbox{(A2)}\quad \kappa := \inf_{d\in\bN}\,\min_{x_i,x_j \in\calX} \,  P^{(d)}_{i,j} (x_i,x_j)>0. \label{eqn:minentry}
\end{align}
That is, assumptions (A1) and (A2) insure that there exists  {\em uniform} lower bounds on the minimum mutual information and the minimum entry in the pairwise probabilities in the forest models as the size of the graph grows. These are typical regularity assumptions for the high-dimensional setting. See \cite{Wai06} and \cite{Mei06} for example. We again emphasize that the proposed learning algorithm  $\CLThres$ has knowledge of neither $I_{\inf}$ nor $\kappa$.   Equipped with (A1) and (A2) and assuming the asymptotic behavior of $\veps_n$ in \eqref{eqn:veps}, we claim the following theorem for $\CLThres$.
\begin{theorem}[Structure Scaling Law] \label{cor:scaling}
There exists two  finite, positive constants $C_1 , C_2$ such that  if 
\begin{equation}
n >\max\left\{ (2\log (d-k))^{1+\zeta},\,  C_1 \log d ,\, C_2 \log k \right\}, \label{eqn:scaling} 
\end{equation}
for any $\zeta>0$, then  the error probability of incorrectly learning the sequence of edge sets $\{E_{P^{(d)}} \}_{d\in\bN}$ tends to zero as $(n,d,k)\to\infty$. When the sequence of forests are trees, $n>C\log d$ (where $C:=\max\{C_1,C_2\}$) suffices for high-dimensional structure recovery. 
\end{theorem}

Thus, if the model parameters $(n,d,k)$ all grow with $n$ but $d=o(\exp( n/C_1 ))$, $k=o(\exp( n/C_2))$ and $d-k=o(\exp(n^{1-\beta}/2))$ (for all $\beta>0$), consistent structure recovery is possible in high dimensions. In other words, the number of nodes $d$ can grow faster than any polynomial  in the sample size $n$. In \cite{Gupta}, the bivariate densities are modeled by functions from a H\"{o}lder class with exponent $\alpha$ and it was mentioned (in Theorem 4.3) that the number of variables can grow like $o(\exp(n^{\alpha/(1+\alpha)}))$ for structural consistency. Our result is somewhat stronger  but we model the pairwise joint distributions as (simpler) probability mass functions (the alphabet $\calX$ is a finite set). 
 
\subsection{Extremal Forest Structures}

In this subsection, we study the extremal structures for learning, that is, the structures that, roughly speaking, lead to the largest and smallest error probabilities for structure learning.  
Define the sequence
\begin{equation}
h_n(P ):=\frac{1}{n \veps_n }\log P^n(\calA_n),\quad \forall\, n\in\bN. \label{eqn:hn}
\end{equation} 
Note that $h_n$ is a function of both the number of variables $d=d_n$ and the number of edges $k=k_n$ in the models $P^{(d)}$ since it is a sequence indexed by $n$. In the next result, we assume $(n,d,k)$ satisfies the scaling law in \eqref{eqn:scaling} and answer the following question: How does $h_n$ in~\eqref{eqn:hn}  depend on the number of edges  $k_n$ for a given $d_n$? Let $P_1^{(d)}$ and  $P_2^{(d)}$  be two sequences of forest-structured distributions with a common number of nodes $d_n$ and  number of edges $k_n( P_1^{(d)})$ and  $k_n( P_2^{(d)})$  respectively.
\begin{corollary}[Extremal Forests] \label{cor:extremal}
Assume that $\CLThres$ is employed as the forest learning algorithm. As $n\to\infty$, $h_n(  P_1^{(d)})\le h_n(  P_2^{(d)})$  whenever $k_n(P_1^{(d)})\ge k_n(P_2^{(d)})$  implying that  $h_n$ is maximized when $P^{(d)}$ are product distributions (i.e., $k_n =0$) and minimized when $P^{(d)}$ are tree-structured distributions (i.e., $k_n =d_n-1$). Furthermore,  if $k_n(P_1^{(d)}) =k_n(P_2^{(d)})$, then $h_n( P_1^{(d)}) = h_n( P_2^{(d)})$. 
\end{corollary}

Note that the corollary is intimately tied to the proposed algorithm $\CLThres.$ We are not claiming that such a result holds for all other forest learning algorithms. The intuition for this result is the following: We recall from the discussion after Theorem~\ref{thm:forest} that the overestimation error dominates the probability of error for structure learning. Thus, the performance of $\CLThres$ degrades with the number of missing edges. If there are very few edges (i.e., $k_n$ is very small relative to $d_n$), the $\CLThres$ estimator is more likely to overestimate the number of edges as compared to if there are many edges (i.e., $k_n/d_n$ is close to 1). We conclude that a distribution which is Markov on an {\em empty graph} (all variables are independent) is the {\em hardest} to learn (in the sense of Corollary~\ref{cor:extremal} above). Conversely,  {\em trees} are the {\em easiest} structures to learn.

\subsection{Lower Bounds on Sample Complexity}
Thus far, our results are for a specific algorithm $\CLThres$ for learning the structure of Markov forest distributions. At this juncture, it is natural to ask whether the scaling laws in Theorem~\ref{cor:scaling} are the best possible over all algorithms (estimators). To answer this question, we limit ourselves to the scenario where the true graph  $T_P$  is a uniformly distributed chance variable\footnote{The term {\em chance variable}, attributed to  \cite{Gallager01}, describes random quantities $Y:\Omega\to W$ that take on values in arbitrary alphabets $W$. In contrast, a random variable $X$ maps the sample space $\Omega$ to the reals $\bR$.  } with probability  measure $\bP$. Assume two different scenarios:
\begin{enumerate}
\item[(a)] $T_P$ is drawn from the uniform distribution on $\calT_k^d$, i.e., $\bP(T_P = t) = 1/|\calT_k^d|$ for all forests $t\in \calT_k^d$. Recall that $ \calT_k^d$ is the set of labeled forests with $d$ nodes and $k$ edges. 
\item[(b)] $T_P$ is drawn from the uniform distribution on $\calF^d$, i.e., $\bP(T_P = t) = 1/|\calF^d|$ for all forests $t\in \calF^d$. Recall that $ \calF^d$ is the set of labeled  forests with $d$ nodes.
\end{enumerate}
This following result is inspired by Theorem 1 in  \cite{Bresler08}. Note that an {\em estimator} or {\em algorithm} $\hT^d$ is simply a map from the set of samples $(\calX^d)^n$ to a set of graphs (either $\calT_k^d$ or $\calF^d$). We emphasize that the following result is stated with the assumption that we are {\em averaging} over the random choice of the true graph $T_P$.
\begin{theorem}[Lower Bounds on Sample Complexity] \label{thm:sample}
Let  $\varrho< 1$ and  $r:=|\calX|$. In case (a) above, if  
\begin{equation}
n< \varrho \, \frac{(k-1) \log d}{d\log r} ,  \label{eqn:sample_lower_bd}
\end{equation}
then $\bP(\hT^d \ne T_P )\to 1$ for any estimator $\hT^d:(\calX^d)^n \to\calT_k^d$. Alternatively, in case (b), if  
\begin{equation}
n< \varrho\, \frac{ \log d}{\log r} ,  \label{eqn:sample_lower_bd_all_forests}
\end{equation}
then $\bP(\hT^d \ne T_P )\to 1$ for any estimator $\hT^d:(\calX^d)^n \to \calF^d$.
\end{theorem}

This result, a {\em strong converse}, states that  $n= \Omega( \frac{k}{d}\log d)$ is {\em necessary} for any  estimator with oracle knowledge of $k$ to succeed. Thus, we need at least logarithmically many samples in $d$ if the fraction $k/d$ is kept constant as the graph size grows even if {\em $k$ is known precisely} and does not have to be estimated. Interestingly, \eqref{eqn:sample_lower_bd} says that if $k$ is large, then we need more samples. This is  because there are fewer forests with a small number of edges as compared to forests with a large number of edges. In contrast, the performance of $\CLThres$ degrades when $k$  is small because it is more sensitive to the overestimation error.  Moreover, if the estimator does not know $k$, then~\eqref{eqn:sample_lower_bd_all_forests} says that $n=\Omega(\log d)$ is {\em necessary} for successful recovery. We conclude that the set of scaling requirements prescribed in Theorem~\ref{cor:scaling} is almost optimal. In fact, if the true structure $T_P$ is a tree,  then Theorem~\ref{thm:sample} for $\CLThres$ says that the (achievability) scaling laws   in Theorem~\ref{cor:scaling} are indeed optimal (up to constant factors in the $O$ and $\Omega$-notation) since $n>(2\log (d-k))^{1+\zeta}$ in \eqref{eqn:scaling} is trivially satisfied. Note that if $T_P$ is a tree, then the Chow-Liu ML procedure or $\CLThres$ results in the  sample complexity  $n=O(\log d)$ (see Theorem~\ref{cor:scaling}).

\section{Risk Consistency} \label{sec:param_consistency}
In this section, we develop results for risk consistency to study how fast the parameters of the estimated distribution converge to their true values. For this purpose, we define the {\em risk} of the estimated distribution $P^*$ (with respect to the true probability model $P$)   as
\begin{equation}
\calR_n (P^*) := D(P \, ||\, P^*) - D(P \, ||\, \tilP), \label{eqn:defrisk}
\end{equation}
where $\tilP$ is the forest projection of $P$ defined in \eqref{eqn:forest_proj}. Note that the original probability model $P$ does not need to be a forest-structured distribution in the definition of the risk. Indeed, if $P$ is Markov on a forest, \eqref{eqn:defrisk} reduces to $\calR_n (P^*) =D(P \, ||\, P^*)$ since the second term is zero.    We quantify the rate of decay of the risk when the number of samples $n$ tends to infinity. For $\delta>0$, we define the event 
\begin{equation}
\calC_{n,\delta}:=\left\{\bx^n\in(\calX^d)^n: \frac{\calR_n (P^*)}{d}  >\delta \right\}. \label{eqn:Cn}
\end{equation}
That is, $\calC_{n,\delta}$ is the event that the {\em average risk}  ${\calR_n (P^*)}/{d}$ exceeds some constant $\delta$. We say that the estimator $P^*$ (or an algorithm) is {\em $\delta$-risk consistent} if the probability of $\calC_{n,\delta}$ tends to zero as $n\to\infty$. Intuitively, achieving $\delta$-risk consistency is easier than  achieving  structural consistency since the learned model $P^*$ can be close to the true forest-projection $\tilP$ in the KL-divergence sense even if their structures differ. 

In order to quantify the rate of decay of the risk in \eqref{eqn:defrisk}, we need to define some  stochastic order notation. We say that a sequence of random variables $Y_n = O_p(g_n)$ (for some deterministic positive sequence $\{g_n\}$) if for every $\epsilon>0$, there exists a $B=B_{\epsilon}>0$ such that $\limsup_{n\to\infty} \Pr(|Y_n|>Bg_n) < \epsilon$. Thus, $\Pr(|Y_n|>Bg_n)\ge\epsilon$ holds for only finitely many $n$.

We say that a reconstruction algorithm has {\em risk consistency of order} (or {\em rate}) $g_n$ if $\calR_n (P^*)  =O_p(g_n)$. The definition of the order of risk consistency involves the true model $P$. Intuitively,  we expect that as $n\to\infty$, the estimated distribution $P^*$ converges to the projection $\tilP$ so $\calR_n (P^*)\to 0$ in probability.

\subsection{Error Exponent for  Risk Consistency}\label{subsec:param_fixed} 
In this subsection, we consider a fixed distribution $P$ and state consistency results in terms of the event $\calC_{n,\delta}$.   Consequently, the model size $d$ and the number of edges $k$ are fixed.  This lends insight into deriving results for the order of the risk consistency and provides intuition for the high-dimensional scenario   in Section~\ref{subsec:param_high}.
\begin{theorem}[Error Exponent for $\delta$-Risk Consistency] \label{thm:consisten_param}
For   $\CLThres$, there exists a constant $\delta_0>0$ such that for all $0<\delta<\delta_0$,
\begin{equation}
\limsup_{n\to\infty}\frac{1}{n}\log P^n(\calC_{n,\delta}) \le  - \delta. \label{eqn:param_upper}
\end{equation}
The corresponding lower bound is 
\begin{equation}
  \liminf_{n\to\infty}\frac{1}{n}\log P^n(\calC_{n,\delta})\ge -\delta \, d. \label{eqn:param_lower} 
\end{equation}
\end{theorem} 
 
The theorem states that if $\delta$ is sufficiently small, the decay rate of the probability of  $\calC_{n,\delta}$ is exponential, hence clearly $\CLThres$ is $\delta$-risk consistent. Furthermore, the bounds on the error exponent associated to the event $\calC_{n,\delta}$ are {\em independent} of the parameters of $P$ and only depend on $\delta$ and the dimensionality $d$. Intuitively, \eqref{eqn:param_upper} is true because if we want the risk of $P^*$   to be at  most $\delta d$, then each of the empirical pairwise marginals $\hP_{i,j}$ should be $\delta$-close to the true pairwise marginal $\tilP_{i,j}$. Note also that for  $\calC_{n,\delta}$ to  occur with high probability, the edge set  does not need to be estimated correctly so there is  no dependence on $k$.

\subsection{The High-Dimensional Setting}  \label{subsec:param_high}
We again consider the high-dimensional setting where the tuple of parameters $(n,d_n,k_n)$ tend to infinity and we have a sequence of learning problems indexed by the number of data points $n$. We again assume that~\eqref{eqn:minMI} and~\eqref{eqn:minentry} hold  and derive sufficient conditions under which the probability of the event $\calC_{n,\delta}$ tends to zero for a sequence of  $d$-variate distributions $\{P^{(d)}\in\calP(\calX^d)\}_{d\in\bN}$. The proof of  Theorem~\ref{thm:consisten_param}  leads immediately to the following corollary.

\begin{corollary}[$\delta$-Risk Consistency Scaling Law] \label{cor:scaling_param}
Let $\delta>0$ be a sufficiently small constant  and  $a\in (0,\delta)$. If the number of variables in the sequence of models  $\{P^{(d)}\}_{d\in\bN}$ satisfies $d_n=o \left(\exp(a n) \right),$ then $\CLThres$ is $\delta$-risk consistent for $\{P^{(d)}\}_{d\in\bN}$.
\end{corollary}

Interestingly, this sufficient condition on how number of variables $d$ should scale with $n$ for consistency is very similar to Theorem~\ref{cor:scaling}. In particular, if $d$ is polynomial in $n$, then $\CLThres$ is both structurally consistent as well as $\delta$-risk consistent. We now study the order of the risk consistency of $\CLThres$ as the model size $d$ grows.

\begin{theorem}[Order of Risk Consistency] \label{cor:estimation_consist}
The risk of the sequence of estimated distributions $\{(P^{(d)})^*\}_{d\in\bN}$ with respect to $\{P^{(d)}\}_{d\in\bN}$   satisfies
\begin{align}
\calR_n((P^{(d)})^*) =  O_p\left(\frac{d\log d}{n^{1-\gamma}}\right), \label{eqn:risk_proj}
\end{align}
for every $\gamma>0$, i.e., the   risk consistency  for $\CLThres$ is of order $(d\log d)/n^{1-\gamma}$. 
\end{theorem}

Note that since this result is stated for the high-dimensional case, $d=d_n$ is a sequence in $n$ but the dependence on $n$ is suppressed for notational simplicity in~\eqref{eqn:risk_proj}. This result implies that if $d = o(n^{1-2\gamma})$  then $\CLThres$ is risk consistent, i.e., $\calR_n((P^{(d)})^*) \to 0$ in probability. Note that this result is not the same as the conclusion of Corollary~\ref{cor:scaling_param} which refers to the probability that the average risk is greater than a fixed constant $\delta$. Also,  the order of convergence given in~\eqref{eqn:risk_proj} does not depend on the true number of edges $k$. This is a consequence of  the result in~\eqref{eqn:param_upper} where the upper bound on the exponent associated to the event $\calC_{n,\delta}$ is independent of the parameters of $P$. 

The order of the risk, or equivalently the rate of convergence of the estimated distribution to the forest projection, is almost linear in the number of variables $d$ and  inversely proportional to $n$. We provide three intuitive reasons to explain why this is plausible:  (i) the dimension of the sufficient statistics in a tree-structured graphical model  is of order $O(d)$  (ii) the ML estimator of the natural parameters of an exponential family converge to their true values at the rate of $O_p(n^{-1/2})$ \citep[Sec.\ 4.2.2]{Serfling} (iii) locally, the KL-divergence behaves like the square of a weighted Euclidean norm of the natural parameters \citep[Eq.\ (11.320)]{Cov06}.

We now compare Theorem~\ref{cor:estimation_consist} to the corresponding results in \cite{Gupta}. In these recent papers, it was shown that by modeling the bivariate densities $\hP_{i,j}$ as functions from a H\"{o}lder class with exponent $\alpha>0$ and using a reconstruction algorithm based on validation on a held-out dataset, the risk decays at a rate\footnote{The $\widetilde{O}_p(\cdot)$ notation suppresses the dependence on factors involving logarithms.} of  $\widetilde{O}_p(  dn^{-\alpha/(1+2\alpha)})$, which is  slower than the order of risk consistency in~\eqref{eqn:risk_proj}. This is due to the need to compute the bivariate densities via kernel density estimation. Furthermore, we model the pairwise joint distributions as discrete probability mass functions and not continuous probability density functions, hence there is no dependence on  H\"{o}lder exponents.

\section{Numerical Results} \label{sec:numerical}
In this section, we perform numerical simulations on synthetic and real datasets to study the effect of a finite number of samples on the probability of  the event $\calA_n$ defined in \eqref{eqn:An}.  Recall that this is the error event associated to an incorrect  learned structure.

\begin{figure}
\centering
\begin{picture}(115,100)
\thicklines
\put(0,50){\line(1,0){100}}
\put(50,0){\line(0,1){50}}
\put(50,50){\line(1,-1){50}}
\put(50,50){\line(-1,-1){50}}
\put(0,50){\circle*{8}}
\put(100,90){\circle*{8}}
\put(100,0){\circle*{8}}
\put(0,90){\circle*{8}}
\put(0,0){\circle*{8}}
\put(50,50){\circle*{8}}
\put(100,50){\circle*{8}}
\put(50,0){\circle*{8}}
\put(50,90){\circle*{8}}
\put(57,60){\makebox (0,0){$X_1$}}
\put(110,60){\makebox (0,0){$X_2$}}
\put(110,10){\makebox (0,0){$X_3$}}
\put(57,10){\makebox (0,0){$X_4$}}
\put(0,10){\makebox (0,0){$X_5$}}
\put(3,60){\makebox (0,0){$X_{k+1}$}}
\put(3,35){\makebox (0,0){$\vdots$}}
\put(3,100){\makebox (0,0){$X_{k+2}$}}
\put(57,100){\makebox (0,0){$X_{k+3}$}}
\put(80,100){\makebox (0,0){$\dots$}}
\put(110,100){\makebox (0,0){$X_{d}$}}
\end{picture} 
 
\caption{The forest-structured distribution Markov on $d$ nodes and $k$ edges. Variables $X_{k+1},\ldots, X_d$ are not connected to the main star graph. }
\label{fig:star}
\end{figure}
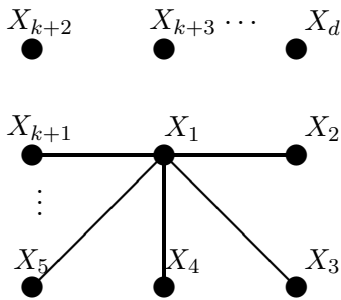 

\subsection{Synthetic Datasets}
In order to compare our estimate to the ground truth graph, we learn the structure of distributions that are Markov on the forest shown in Figure~\ref{fig:star}. Thus, a subgraph (nodes $1,\ldots, k+1$) is a (connected) star while nodes $k+2,\ldots, d-1$ are not connected to the star. Each random variable $X_j$ takes on values from a binary alphabet $\calX=\{0,1\}$. Furthermore, $P_j(x_j)=0.5$ for $x_j=0,1$ and all $j\in V$. The conditional distributions are governed by the ``binary symmetric channel'':
\begin{eqnarray}
P_{j|1}(x_j|x_1) = \left\{ \begin{array}{cc}
0.7 & x_j = x_1\\
0.3 & x_j \ne x_1\\
\end{array}\right.      \nn 
\end{eqnarray}
for $j=2,\ldots, k+1$. We further assume that the regularization sequence is given by  $\veps_n := n^{-\beta}$ for some $\beta\in (0,1)$. Recall that this  sequence satisfies the conditions in \eqref{eqn:veps}. We will vary $\beta$ in our experiments to observe its effect on the overestimation and underestimation errors.

\begin{figure}
\centering
\includegraphics[width=.475\columnwidth]{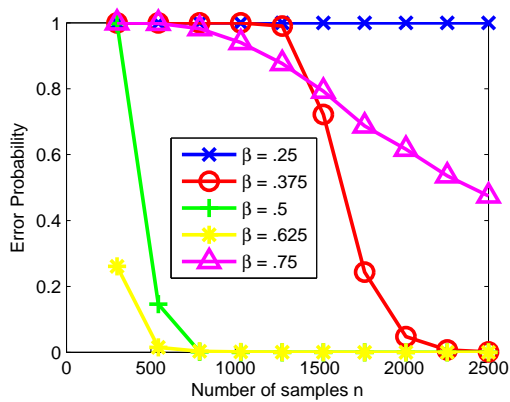}
\caption{The error probability  of structure learning for  $\beta\in (0,1)$. }
\label{fig:struct_error_prob}
\end{figure}

In Figure~\ref{fig:struct_error_prob}, we show  the simulated error probability as a function of the sample size $n$ for a  $d=101$ node graph (as in Figure~\ref{fig:star}) with $k=50$ edges. The error probability is estimated based on 30,000 independent runs of $\CLThres$ (over different datasets $\bx^n$). We observe that the error probability is minimized when  $\beta\approx 0.625$. Figure~\ref{fig:underest} show the simulated overestimation and underestimation errors for this experiment. We see that as $\beta\to 0$, the overestimation (resp.\ underestimation) error is likely to be small (resp.\ large) because the regularization sequence $\veps_n$ is large.  When the number of samples is relatively small as  in this experiment, both types of errors contribute significantly to the overall error probability.  When $\beta\approx 0.625$, we have the best tradeoff between overestimation and underestimation for this particular experimental setting.

Even though we mentioned that $\beta$ in \eqref{eqn:vepsbeta} should be chosen to be close to zero so that the error probability of structure learning decays as rapidly as possible, this example demonstrates that when given a finite number of samples, $\beta$ should  be chosen  to balance the overestimation and underestimation errors.    This does not violate Theorem~\ref{thm:forest} since  Theorem~\ref{thm:forest} is an asymptotic result and refers to the typical way an error occurs in the limit as $n\to\infty$. Indeed, when the number of samples is very large, it is shown that the overestimation error dominates the overall probability of error and so one should choose $\beta$ to be close to zero. The question of how best to select optimal $\beta$ when given only a finite number of samples appears to be a challenging one. We use cross-validation as a proxy to select this parameter for the real-world datasets in the next section. 

\begin{figure}
\centering
\begin{tabular}{cc}
\includegraphics[width=.475\columnwidth]{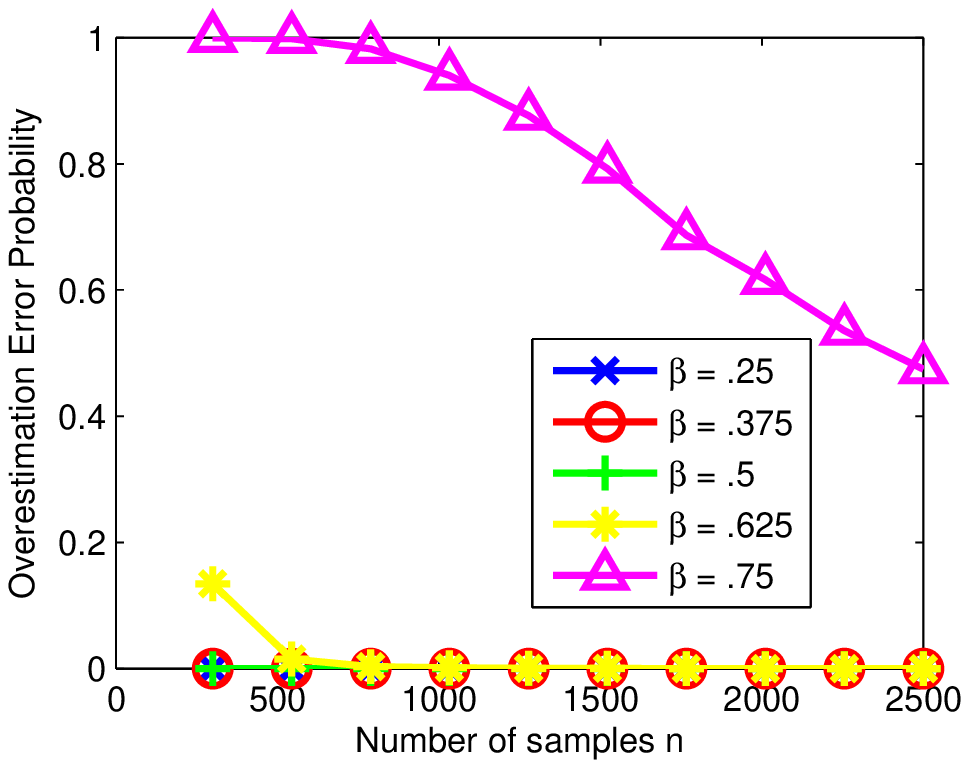} &
\includegraphics[width=.475\columnwidth]{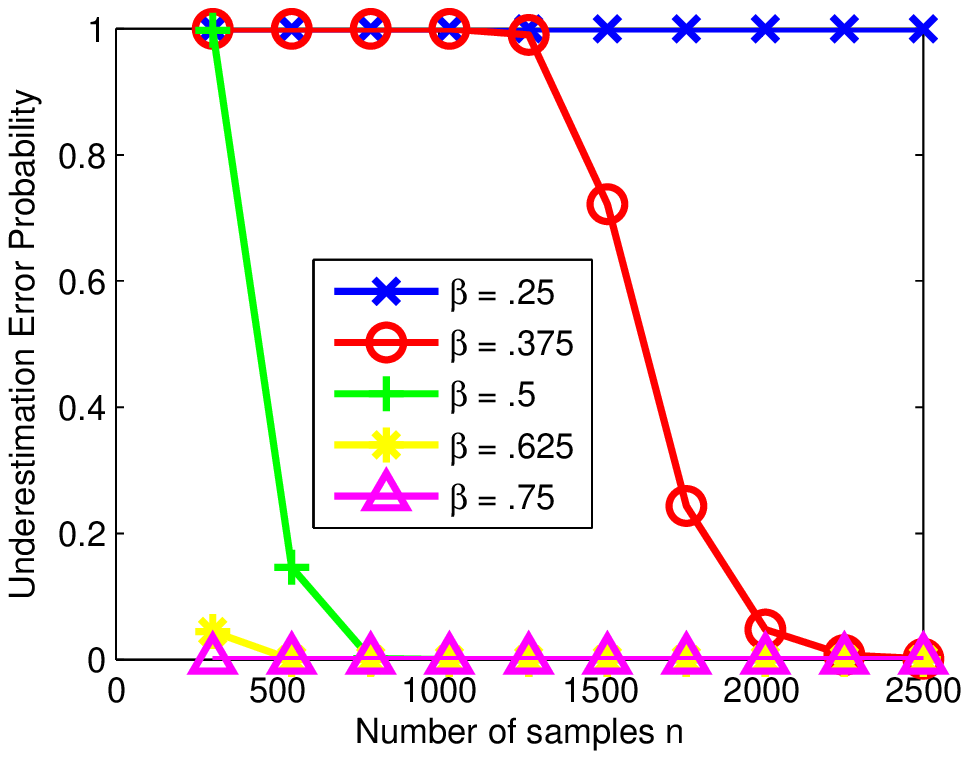}
\end{tabular}
\caption{The overestimation and underestimation errors  for  $\beta\in (0,1)$. }
\label{fig:underest}
\end{figure}

In Figure~\ref{fig:params}, we fix the value of $\beta$ at $0.625$ and plot the KL-divergence  $D(P \, ||\, P^*)$ as a function of the number of samples. This is done  for a forest-structured distribution $P$ whose graph is shown in Figure~\ref{fig:star} and with $d=21$ nodes and $k=10$ edges. The mean, minimum and maximum KL-divergences are computed based on 50 independent runs of $\CLThres$. We see that $\log D(P \, ||\, P^*)$ decays linearly. Furthermore, the slope of the mean curve  is approximately $-1$, which is in agreement with  \eqref{eqn:risk_proj}. This experiment shows that if we want to reduce the KL-divergence between the estimated and true models by a constant factor $A>0$, we need to increase the number of samples by roughly the same factor $A$. 

\begin{figure}
\centering
\includegraphics[width=.55\columnwidth]{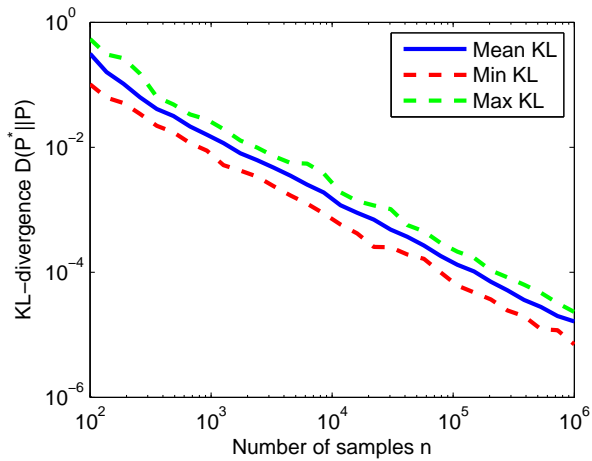}
\caption{Mean, minimum and maximum (across 50 different runs) of the KL-divergence   between the estimated model $P^*$ and the true model $P$ for a $d=21$ node graph with $k=10$ edges. }
\label{fig:params}
\end{figure}

\subsection{Real datasets}
We now demonstrate how well forests-structured distributions can model  two real datasets\footnote{These datasets are typically employed for binary classification but we use them for modeling purposes.} which are obtained from the UCI Machine Learning Repository \citep{UCI98}. The first dataset we used is known as the SPECT Heart dataset, which describes diagnosing of cardiac Single Proton Emission Computed Tomography (SPECT) images on normal and abnormal patients. The dataset contains $d=22$ binary variables and $n=80$ training samples. There are also 183 test samples. We learned a forest-structured distributions using  the 80 training samples for different $\beta\in (0,1)$ and subsequently computed the log-likelihood of both  the training and test samples. The results are displayed in Figure~\ref{fig:spect}. We observe that, as expected, the log-likelihood of the training samples increases monotonically with $\beta$. This is because there are more edges in the model when $\beta$ is large improving the modeling ability. However, we observe that there is overfitting when $\beta$ is large as evidenced by the decrease in the log-likelihood of the 183 test samples. The optimal value of $\beta$ in terms of the log-likelihood for this dataset is $\approx 0.25,$ but surprisingly an approximation with an empty graph\footnote{When $\beta=0$ we have an empty graph because all empirical mutual information quantities in this experiment are smaller than 1.} also yields a high log-likelihood score on the test samples. This implies that according to the available data, the variables are nearly independent. The forest graph for $\beta=0.25$ is shown in Figure~\ref{fig:spect_graph}(a) and is very sparse.

\begin{figure}
\centering
\begin{tabular}{cc}
\includegraphics[width=.475\columnwidth]{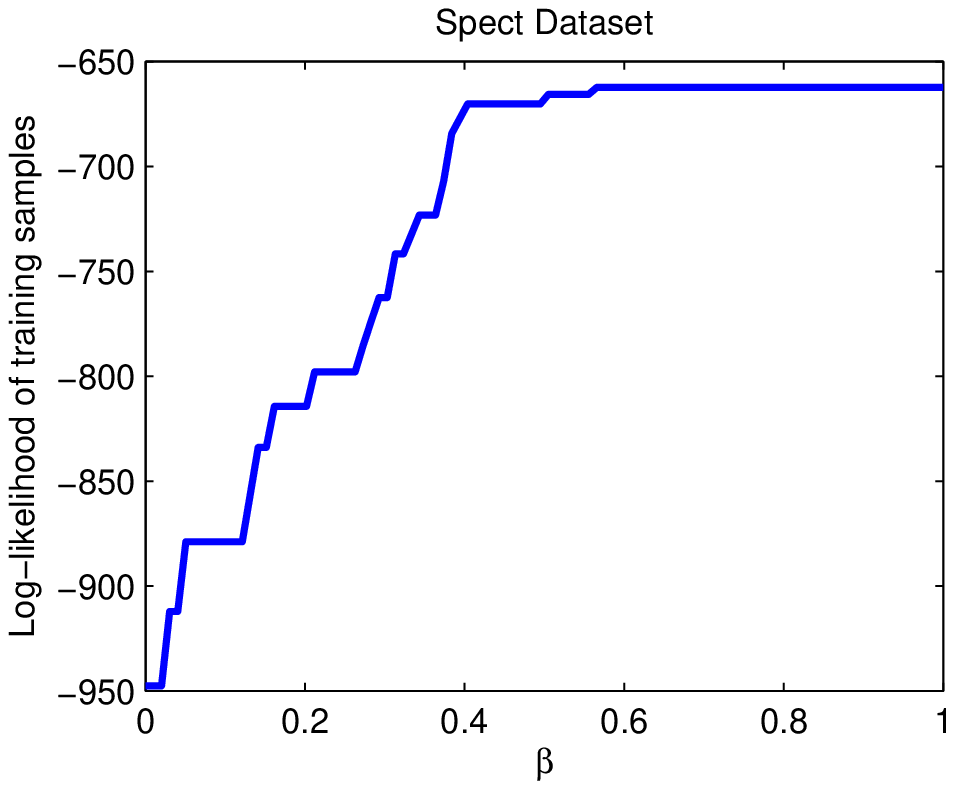} &
\includegraphics[width=.475\columnwidth]{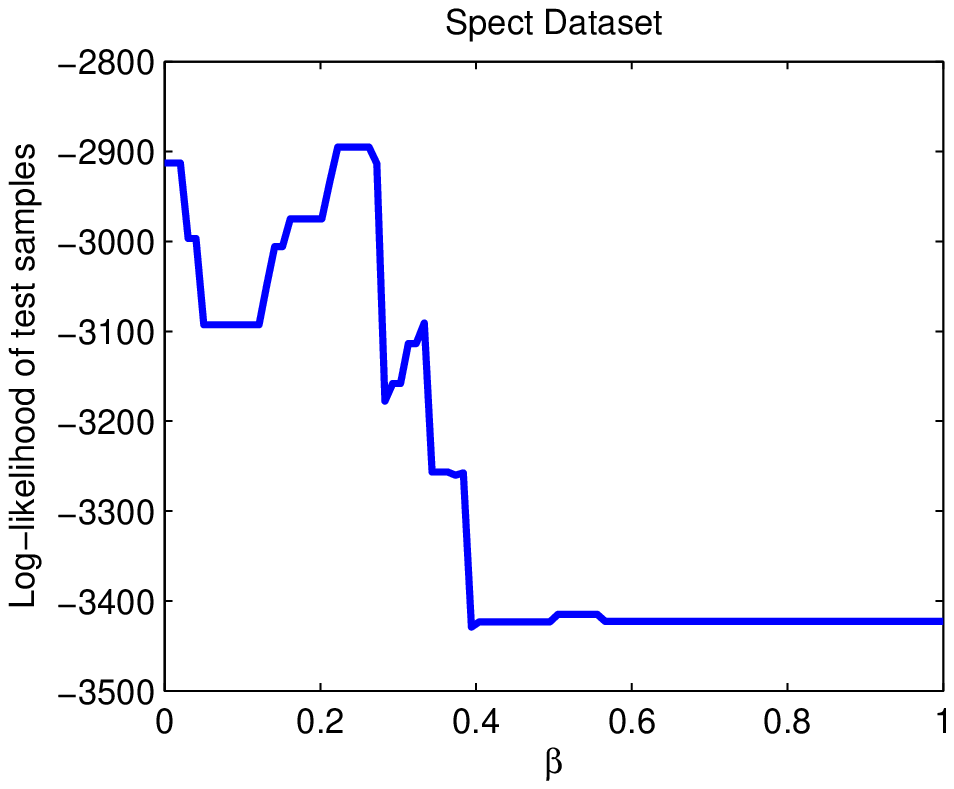}
\end{tabular}
\caption{Log-likelihood scores on the SPECT dataset}
\label{fig:spect}
\end{figure}

The second dataset we used is the Statlog Heart dataset containing physiological measurements of subjects with and without heart disease. There are  $270$ subjects and $d=13$ discrete and continuous attributes, such as gender and resting blood pressure. We quantized the continuous attributes into two bins. Those measurements that are above the mean are encoded as 1 and those below the mean as 0. Since the raw dataset is not partitioned into training and test sets, we learned  forest-structured models based on  a randomly chosen set of $n=230$ training samples and then computed the log-likelihood of these training and 40 remaining test samples.  We then chose an additional 49  randomly partitioned training and test sets and performed the same learning task and computation of log-likelihood scores. The mean of the log-likelihood scores over these 50 runs is shown in Figure~\ref{fig:heart}. We observe that the log-likelihood on the test set is maximized at $\beta\approx 0.53$ and the tree approximation ($\beta\approx 1$) also yields a high likelihood score.  The forest  learned  when $\beta = 0.53$ is shown in Figure~\ref{fig:spect_graph}(b). Observe that two nodes (ECG and Cholesterol) are disconnected from the main graph because their mutual information values with other variables are below the threshold. In contrast, HeartDisease, the label for this dataset, has the highest degree, i.e., it influences and is influenced by many other covariates. The strengths of the interactions between HeartDisease and its neighbors are also strong as evidenced by the bold edges.

From these experiments, we  observe that  some datasets can be modeled well as proper forests with very few edges while others are better modeled as distributions that are almost tree-structured (see Figure~\ref{fig:spect_graph}). Also, we need to choose $\beta$ carefully to balance between data fidelity and overfitting.   In contrast,  our asymptotic result in Theorem~\ref{thm:forest} says that $\veps_n$ should be chosen according to \eqref{eqn:veps} so that we have structural consistency. When the number of data points $n$ is large, $\beta$ in~\eqref{eqn:vepsbeta} should be chosen to be small to ensure that the learned edge set is equal to the true one (assuming the underlying model is a forest) with high probability as the overestimation error dominates.  

\begin{figure}
\centering
\begin{tabular}{cc}
\includegraphics[width=0.4\columnwidth]{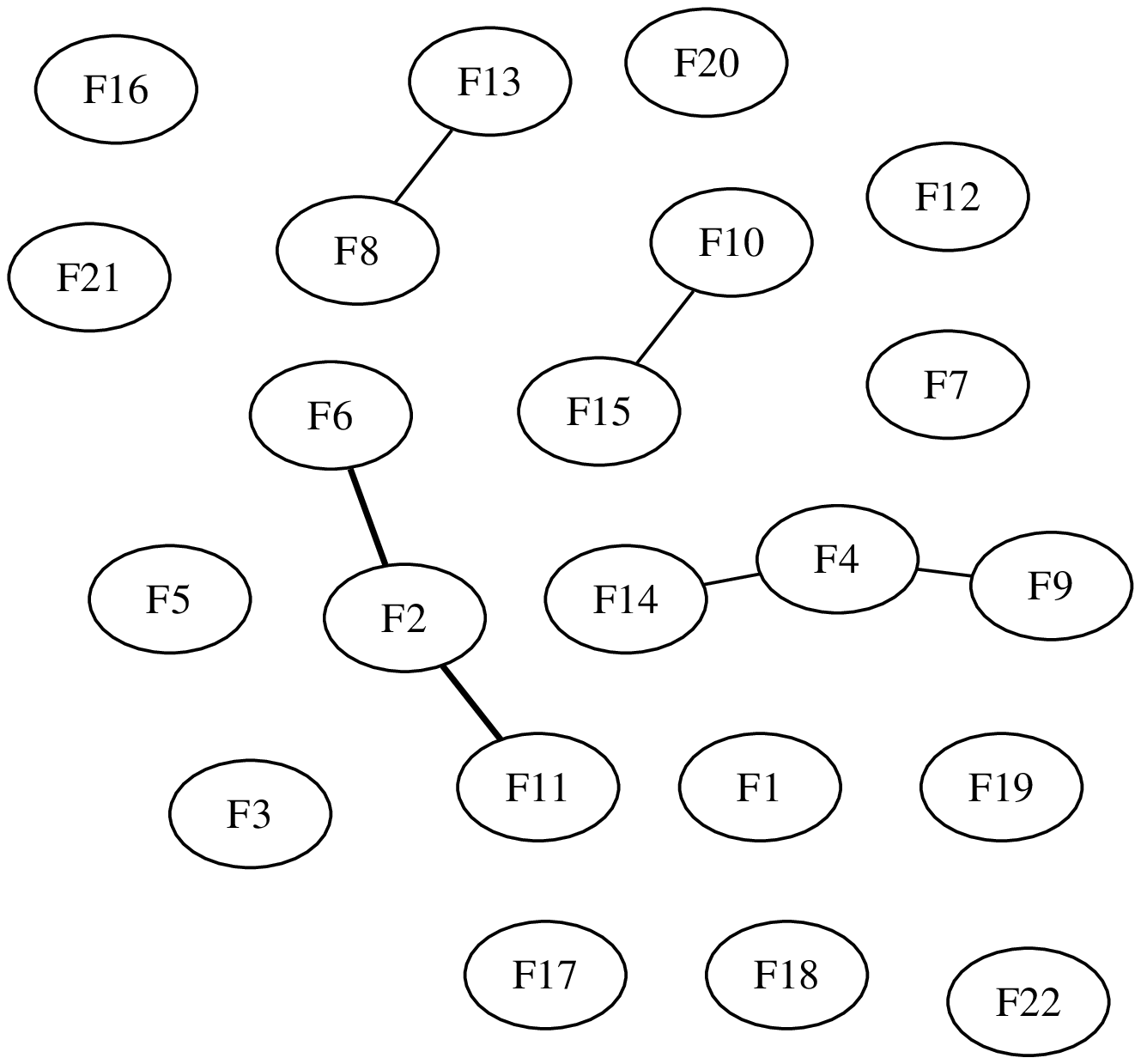} & 
\includegraphics[width=0.6\columnwidth]{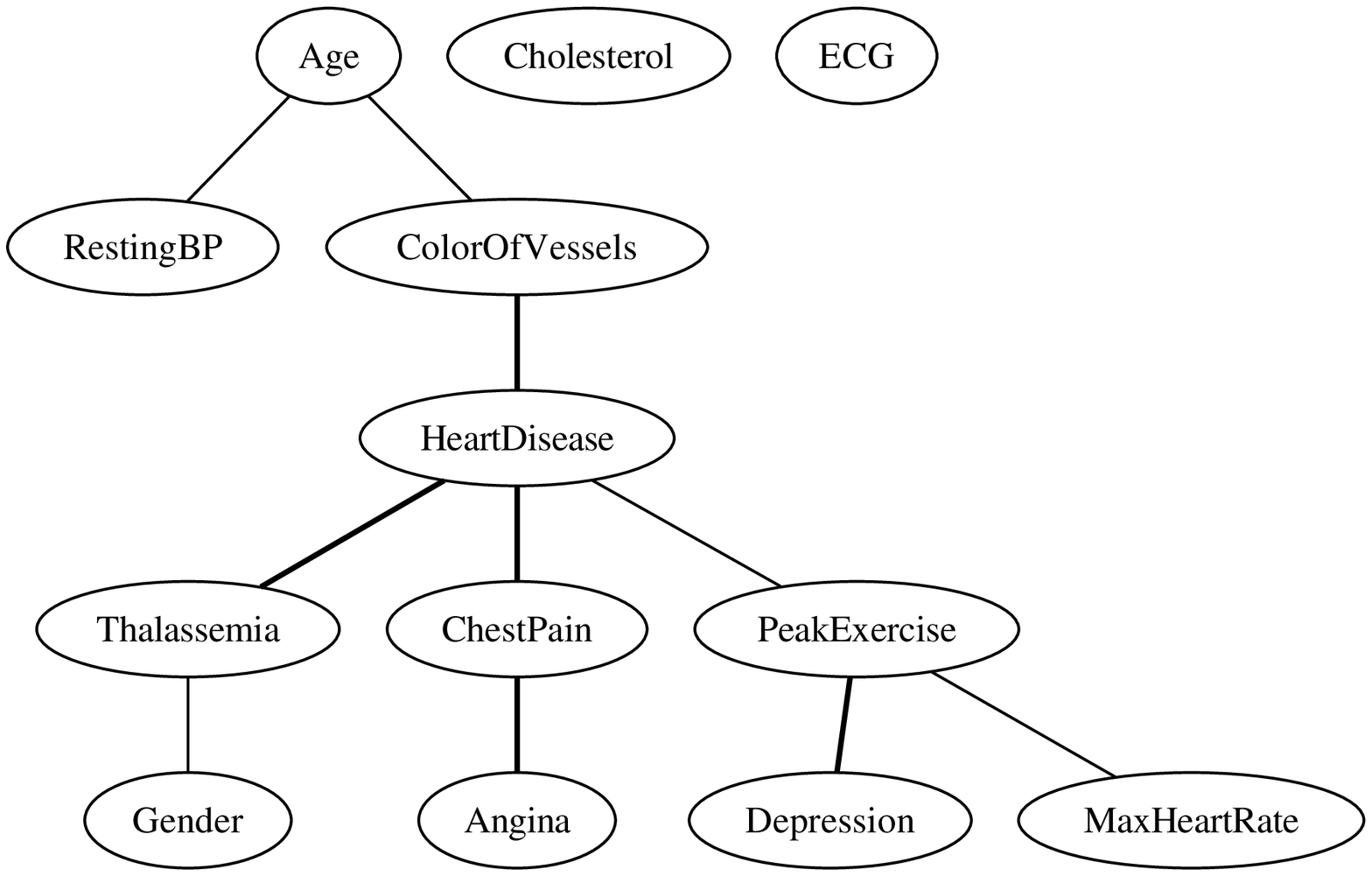} \\
(a) & (b)
\end{tabular}
\caption{Learned forest graph of the (a) SPECT  dataset for $\beta=0.25$ and (b) HEART dataset for $\beta=0.53$. Bold edges denote higher mutual information values. The features names are not provided for the SPECT  dataset.  }
\label{fig:spect_graph}
\end{figure}
\begin{figure}
\centering
\begin{tabular}{cc}
\includegraphics[width=.475\columnwidth]{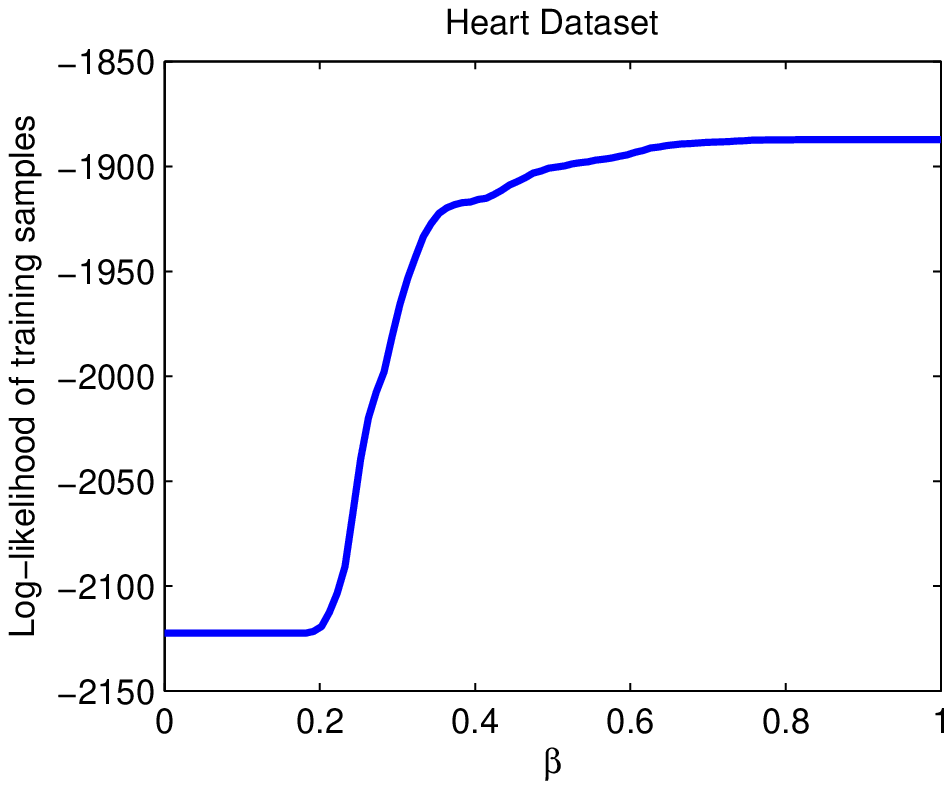} &
\includegraphics[width=.475\columnwidth]{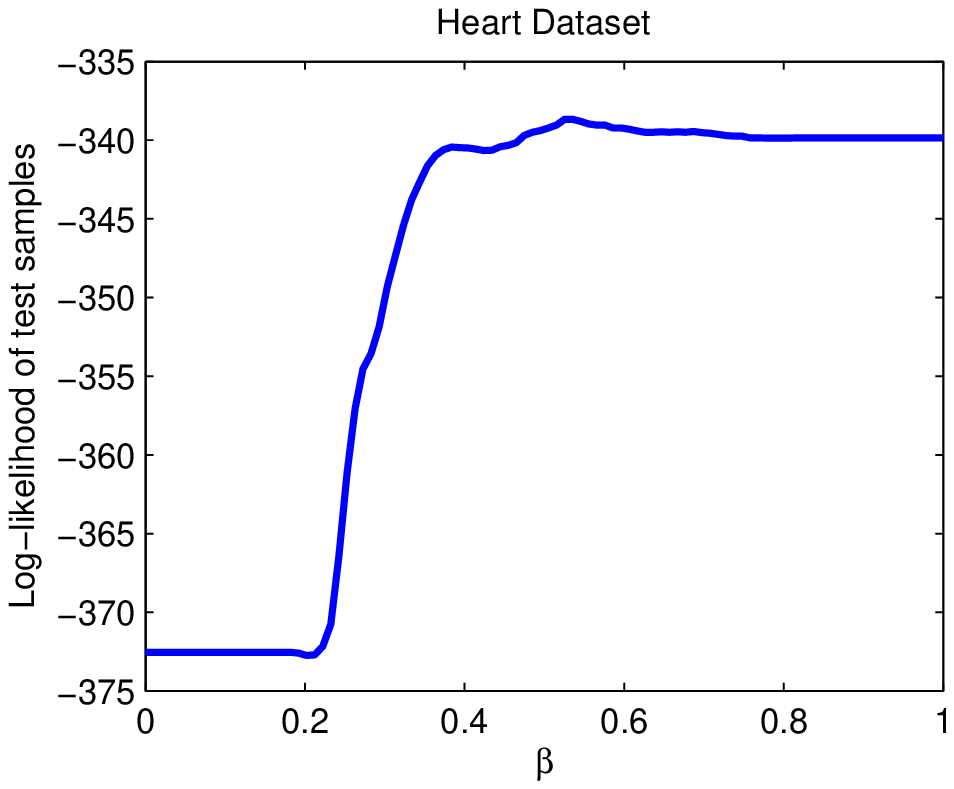}
\end{tabular}
\caption{Log-likelihood scores on the HEART dataset}
\label{fig:heart}
\end{figure}

\section{Conclusion} \label{sec:conclusion}
In this paper, we  proposed an efficient algorithm $\CLThres$ for learning the parameters and the structure of forest-structured graphical models. We showed that the asymptotic error rates associated to structure learning are nearly optimal. We also provided the rate at which the error probability of structure learning tends to zero and the order of the risk consistency. One natural question that arises from our analyses is whether $\beta$ in \eqref{eqn:vepsbeta} can be selected automatically in the finite-sample regime.   There are many other open problems that could possibly leverage on the proof techniques employed here.  For example, we are currently interested to analyze  the learning of general graphical models  using similar thresholding-like techniques on the empirical correlation coefficients. The analyses could potentially leverage on the use of the method of types. We are currently exploring this promising line of research.

\acks{This work was supported by a AFOSR funded through Grant FA9559-08-1-1080,   a MURI funded through ARO Grant W911NF-06-1-0076 and a MURI funded
through AFOSR Grant FA9550-06-1-0324. V.~Tan is also funded by  A*STAR, Singapore. The authors would like to thank   Sanjoy Mitter, Lav Varshney, Matt Johnson and James Saunderson for discussions.}



\appendix


\section{Proof of Proposition \ref{prop:Imin}} \label{prf:prop:Imin}
\begin{proof}
({\it Sketch}) The proof of this result hinges on the fact that both the overestimation and underestimation errors decay to zero exponentially fast when the threshold is chosen to be $I_{\min}/2$. This threshold is able to differentiate between true edges (with MI larger than $I_{\min}$)  from non-edges  (with MI smaller than $I_{\min}$) with high probability for $n$ sufficiently large. The error for learning the top $k$ edges of the forest also decays exponentially fast \citep{Tan&etal:09ITsub}. Thus, \eqref{eqn:expo_Imin} holds. The full details of the proof follow in a straightforward manner from  Appendix~\ref{prf:thm:forest} which we present next.
\end{proof}

\section{Proof of Theorem \ref{thm:forest}} \label{prf:thm:forest}
Define the event $\calB_n := \{\hE_k \ne E_P\},$ where $\hE_k = \{\he_1,\ldots,\he_k\}$ is the set of top $k$ edges (see Step~\ref{item:mwst} of $\CLThres$ for notation). This is the Chow-Liu error as mentioned in Section~\ref{sec:proofidea}. Let $\calB_n^c$ denote the complement of $\calB_n$. Note that in $\calB_n^c$, the estimated edge set depends on  $k$, the true model order, which is {\em a-priori} unknown to the learner.  Further define the  constant
\begin{equation}
K_P:=  \lim_{n\to\infty} -\frac{1}{n}\log P^n(\calB_n).  \label{eqn:KPdef}
\end{equation}
In other words, $K_P$ is the error exponent for learning the forest structure incorrectly assuming the true model order $k$ is known and Chow-Liu terminates after the addition of exactly $k$ edges in the MWST procedure \citep{Kruskal}.   The existence of the limit in~\eqref{eqn:KPdef} and the positivity of $K_P$ follow from the main results in~\cite{Tan&etal:09ITsub}. 

We first state a result which relies on the Gallager-–Fano bound~\cite[pp.\ 24]{Fano}. The proof will be provided at the end of this appendix. 
\begin{lemma}[Reduction to Model Order Estimation] \label{lem:reduction}
For every $\eta\in (0,K_P)$, there exists a $N\in\bN$ sufficiently large such that for every $n>N$, the error probability $P^n(\calA_n)$ satisfies
\begin{align}
(1-\eta)P^n & (\hk_n\ne k|\calB_n^c) \le P^n(\calA_n) \label{eqn:reduction_lower}\\
&\le P^n(\hk_n\ne k|\calB_n^c) + 2\exp(-n (K_P-\eta)). \label{eqn:reduction}
\end{align}
\end{lemma}

\begin{proof} ({\it of Theorem~\ref{thm:forest}})
We will prove (i) the upper bound in~\eqref{eqn:limsup_statement}  (ii) the lower bound in~\eqref{eqn:liminf_statement} and (iii) the exponential rate of decay in the case of trees~\eqref{eqn:exponentially_fast}. 
\subsection*{Proof of upper bound in Theorem~\ref{thm:forest}}
We now bound the error probability $P^n(\hk_n\ne k|\calB_n^c)$ in~\eqref{eqn:reduction}. Using  the union bound,
\begin{equation}
P^n(\hk_n\ne k|\calB_n^c)\le P^n(\hk_n > k|\calB_n^c)+ P^n(\hk_n < k|\calB_n^c). \label{eqn:over_under}
\end{equation}
The first and second terms are known as the {\em overestimation} and {\em underestimation} errors respectively. We will show that the underestimation error decays exponentially fast. The overestimation error decays only subexponentially fast and so its rate of decay dominates the overall rate of decay of the error probability for structure learning.

\subsubsection*{Underestimation Error}
We now bound these terms staring with the underestimation error. By the union bound, 
\begin{align}
P^n(\hk_n < k|\calB_n^c)  &\le (k-1)\max_{1\le j \le k-1} P^n(\hk_n =j|\calB_n^c)\nn \\
&=(k-1)P^n(\hk_n =k-1|\calB_n^c), \label{eqn:under_1}
\end{align} 
where \eqref{eqn:under_1} follows because $P^n(\hk_n  = j|\calB_n^c)$ is maximized when $j = k - 1$. This is because if, to the contrary, $P^n(\hk_n  = j|\calB_n^c)$ were to be maximized at some other $j\le k-2$, then there exists at least two edges, call them  $e_1,e_2\in E_P$ such that events  $\calE_1:=\{I(\hP_{e_1})\le \veps_n\}$ and $\calE_2:=\{I(\hP_{e_2})\le \veps_n\}$ occur. The probability of this joint event is  smaller than the individual probabilities, i.e.,  $P^n(\calE_1\cap\calE_2)\le\min\{P^n(\calE_1),P^n(\calE_2)\}$. This is a contradiction. 

By the rule for choosing $\hk_n$ in~\eqref{eqn:kn},  we have the upper bound
\begin{align}
P^n(\hk_n =k-1|\calB_n^c) &= P^n(\exists\, e\in E_P \mbox{ s.t. } I(\hP_e) \le \veps_n)\le k\max_{e\in E_P} P^n(  I(\hP_e) \le \veps_n), \label{eqn:under_2}
\end{align}
where the inequality follows from the union bound. Now, note  that if $e\in E_P$, then $I(P_e) >\veps_n$ for $n$ sufficiently large (since $\veps_n\to 0$). Thus, by Sanov's theorem \cite[Ch.\ 11]{Cov06},  $P^n(I(\hP_e) \le \veps_n)$ can be  upper bounded as 
\begin{eqnarray}
P^n(I(\hP_e) \le \veps_n)\le(n +  1)^{r^2}  \exp\left(-n \min_{Q\in\calP(\calX^2)}  \left\{ D(Q\, ||\, P_e) : I(Q) \le   \veps_n \right\} \right). \label{eqn:sanov_under}
\end{eqnarray}
Define the good rate function \citep{Dembo} in~\eqref{eqn:sanov_under} to be  $L :\calP(\calX^2)\times [0,\infty) \to [0,\infty)$, which is given by
\begin{equation}
L (P_e;a ) :=  \min_{Q\in\calP(\calX^2)} \left\{ D(Q\, ||\, P_e) : I(Q)\le a \right\} . \label{eqn:LPe}
\end{equation}
Clearly, $L(P_e;a)$ is continuous in $a$. Furthermore it is monotonically decreasing in $a$ for fixed $P_e$. Thus by using the continuity of $L(P_e;\cdot)$ we can assert: To every $\eta>0$, there exists a $N\in\bN$ such that for all $n>N$ we have $L (P_e;\veps_n) >   L(P_e; 0)-\eta$. As such,  we can further upper bound the  error probability in~\eqref{eqn:sanov_under}  as 
\begin{equation}
 P^n(  I(\hP_e) \le \veps_n) \le (n+1)^{r^2} \exp\left(-n (L(P_e; 0)-\eta) \right).  \label{eqn:under_3}
\end{equation}
By using the fact that $I_{\min}>0$, the exponent $L(P_e; 0)>0$ and thus, we can put the pieces in \eqref{eqn:under_1}, \eqref{eqn:under_2} and \eqref{eqn:under_3} together to show that the underestimation error is upper bounded as
\begin{align}
P^n(\hk_n < k|\calB_n^c) \le k(k-1) (n+1)^{r^2}\exp\left(-n \min_{e\in E_P} (L(P_e; 0) -\eta)\right). \label{eqn:under_4}
\end{align} 
Hence, if $k$ is constant, the underestimation error $P^n(\hk_n < k |\calB_n^c)$ decays to zero exponentially fast as $n\to\infty$, i.e, the normalized logarithm of the underestimation error can be bounded as 
\begin{equation}
\limsup_{n\to\infty} \frac{1}{n}\log P^n(\hk_n < k|\calB_n^c) \le -\min_{e\in E_P} (L(P_e; 0) - \eta).\nn \end{equation}
The above statement is now independent of $n$. Hence, we can take the limit  as $\eta\to 0$ to conclude that:
\begin{equation}
\limsup_{n\to\infty} \frac{1}{n}\log P^n(\hk_n < k|\calB_n^c) \le  -L_P. \label{eqn:under_error}
\end{equation}
The exponent $L_P:=\min_{e\in E_P}L(P_e;0)$ is positive because we assumed that the model is minimal and so $I_{\min}>0$, which ensures the positivity of the rate function $L(P_e; 0)$ for each true edge $e\in E_P$.

\subsubsection*{Overestimation Error}
Bounding the overestimation error is harder. It follows by first applying the union bound:
\begin{align}
P^n(\hk_n > k|\calB_n^c) & \le (d-k-1)\max_{k+1\le j \le d-1} P^n(\hk_n =j|\calB_n^c) \nn \\
&=(d-k-1)P^n(\hk_n =k+1|\calB_n^c), \label{eqn:over_1}
\end{align}
where \eqref{eqn:over_1} follows because  $P^n(\hk_n =j|\calB_n^c)$ is maximized when $j=k+1$ (by the same argument as for the underestimation error). Apply the union bound again, we have 
\begin{eqnarray}
P^n(\hk_n  =  k+1|\calB_n^c) \le  (d - k - 1) \max_{e\in V\times V : I(P_e) = 0} P^n( I (\hP_e)  \ge  \veps_n) . \label{eqn:over_2}
\end{eqnarray}
From \eqref{eqn:over_2}, it suffices to  bound $P^n( I (\hP_e) \ge \veps_n)$ for any pair of independent random variables $(X_i,X_j)$ and $e=(i,j)$. We proceed by applying the upper bound in Sanov's theorem \citep[Ch.\ 11]{Cov06} to $P^n( I (\hP_e) \ge \veps_n)$ which yields 
\begin{eqnarray}
P^n( I (\hP_e) \ge \veps_n)\le (n+1)^{r^2} \exp\left(-n \min_{Q\in\calP(\calX^2) } \left\{ D(Q\, ||\, P_e) : I(Q) \ge \veps_n\right\}  \right),  \label{eqn:over_3}
\end{eqnarray}
for all $n\in\bN$. Our task now is to lower bound the good rate function in~\eqref{eqn:over_3}, which we denote as $M:\calP(\calX^2)\times [0,\infty)\to[0,\infty)$:
\begin{equation}
M(P_e; b):= \min_{Q\in\calP(\calX^2) } \left\{ D(Q\, ||\, P_e) : I(Q) \ge b\right\} . \label{eqn:upper_exponent}
\end{equation}
Note that $M(P_e; b)$ is monotonically increasing and continuous in $b$ for fixed $P_e$. Because the sequence $\{\veps_n\}_{n\in\bN}$ tends to zero, when $n$ is sufficiently large, $\veps_n$ is arbitrarily small and we are in the so-called {\em very-noisy learning regime} \citep{Bor08, Tan&etal:09ITsub}, where the optimizer to~\eqref{eqn:upper_exponent}, denoted as $Q_n^*$, is very close to $P_e$. See Figure~\ref{fig:manifold_illus}. 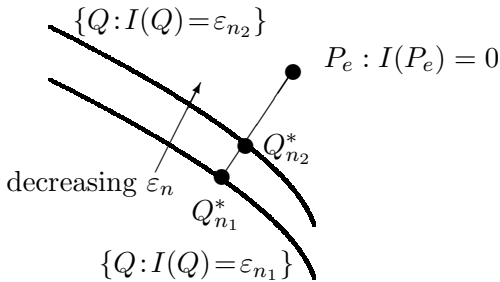
\begin{figure}
\centering
\begin{picture}(100,90)
\linethickness{1pt}
\qbezier(0,90)(90,45)(100,15)
\qbezier(0,70)(90,25)(100,-5)
\put(40,35){\vector(1,2){17}}
\put(100,75){\mbox{ $P_e: I(P_e)=0$}}
\put(15,-3){\mbox{ $\{Q\!:\! I(Q) \!=\! \veps_{n_1}\}$}}
\put(5,89){\mbox{ $\{Q\!:\! I(Q)\!=\!\veps_{n_2}\}$}}
\put(-20,28){\mbox{ decreasing $\veps_n$}}
\put(50,17){\mbox{ $Q_{n_1}^*$}}
\put(77,42){\mbox{ $Q_{n_2}^*$}}
\linethickness{.75pt}
\put(65,33){\line(2,3){28}}
\put(65,33){\circle*{6}}
\put(74,45){\circle*{6}}
\put(92,73){\circle*{6}}
\end{picture}
\caption{As $\veps_n\to 0$, the projection of $P_e$ onto the constraint set $\{Q :I(Q)\ge\veps_n\}$, denoted $Q_n^*$  (the optimizer in \eqref{eqn:upper_exponent}), approaches $P_e$. The approximations in~\eqref{eqn:dpq_approximations} and~\eqref{eqn:approximations} become increasingly accurate as $\veps_n$ tends to zero. In the figure, $n_2>n_1$ and $\veps_{n_1}>\veps_{n_2}$ and the curves are the (sub-)manifold of distributions such that the mutual information is constant, i.e., the mutual information level sets.}
\label{fig:manifold_illus}
\end{figure}
 Thus, when  $n$ is large, the KL-divergence and mutual information can be approximated  as
\begin{align}
D(Q_n^*\,  ||\, P_e) &= \frac{1}{2} \bv^T \bPi_e \bv + o(\|\bv\|^2), \label{eqn:dpq_approximations} \\*
I(Q_n^*) &= \frac{1}{2}  \bv^T \bH_e \bv + o(\|\bv\|^2),  \label{eqn:approximations}
\end{align}
where\footnote{The operator $\vect(\bC)$ vectorizes a matrix in a column oriented way. Thus, if $\bC\in\bR^{l\times l}$, $\vect(\bC)$ is a length-$l^2$ vector with the columns of $\bC$ stacked one on top of another ($\bC(:)$ in Matlab). } $\bv:=\vect(Q_n^*)-\vect(P_e)\in\bR^{r^2}$.  The  $r^2\times r^2$ matrices $\bPi_e$ and $ \bH_e$ are defined as
\begin{align}
\bPi_e &:= \diag(1/\vect(P_e)), \label{eqn:Pi_e}\\
\bH_e  &:= \nabla_{\vect(Q)}^2 I(\vect(Q)) \big|_{Q=P_e}.   \label{eqn:H_e}
\end{align}
In other words, $\bPi_e$ is the diagonal matrix that contains the reciprocal of the elements of $\vect(P_e)$ on its diagonal. $\bH_e$ is the Hessian\footnote{The first two terms in the Taylor  expansion of the mutual information $I(\vect(Q_n^*))$ in~\eqref{eqn:approximations} vanish because (i) $I(P_e)=0$ and (ii) $(\vect(Q_n^*)-\vect(P_e))^T \nabla_{\vect(Q)}  I(\vect(P_e)) =0$. Indeed, if we expand $I(\vect(Q))$ around a product distribution, the constant and linear terms vanish \citep{Bor08}. Note that $\bH_e$ in \eqref{eqn:H_e} is an indefinite matrix because $I(\vect(Q))$ is not convex. }   of   $I(\vect(Q_n^*))$, viewed as a function of $\vect(Q_n^*)$ and evaluated at $P_e$. As such, the exponent for overestimation in~\eqref{eqn:upper_exponent} can be approximated by  a {\em quadratically constrained quadratic program} (QCQP), where $\bz:=\vect(Q)-\vect(P_e)$:
\begin{align}
\tilM(P_e; \veps_n) &= \min_{\bz\in\bR^{r^2} }   \frac{1}{2} \bz^T \bPi_e \bz ,\nn\\ 
\mbox{subject to}\quad  \frac{1}{2}  \bz^T &\bH_e \bz \ge \veps_n, \quad \bz^T  \mathbf{1}= 0 . \label{eqn:primal2}
\end{align} 
Note that the constraint $\bz^T  \mathbf{1}= 0$ does not necessarily ensure that $Q$ is a probability distribution so $\tilM(P_e; \veps_n) $ is a lower bound to the true rate function $M(P_e; \veps_n)$, defined in~\eqref{eqn:upper_exponent}.  We now argue that the approximate rate function $\tilM$ in~\eqref{eqn:primal2}, can be lower bounded by a quantity that is  proportional to $\veps_n$.  
To show this, we resort to Lagrangian duality \citep[Ch.\ 5]{Ber99}.  It can easily be shown that the {\em Lagrangian dual}  corresponding to the primal in~\eqref{eqn:primal2} is
\begin{align}
g(P_e; \veps_n) 
:=  \veps_n\max_{\mu\ge 0}  \{  \mu : \bPi_e \succeq\mu  \bH_e\}. \label{eqn:sdp_opt}
\end{align}
We see from \eqref{eqn:sdp_opt} that $g(P_e; \veps_n)$ is proportional to $\veps_n$. By weak duality \citep[Proposition 5.1.3]{Ber99}, any dual feasible solution provides a lower bound to the primal, i.e., 
\begin{equation}
g(P_e; \veps_n)\le \tilM(P_e; \veps_n). \label{eqn:weak_duality}
\end{equation}
Note that strong duality (equality in \eqref{eqn:weak_duality}) does not hold in general due in part to the non-convex constraint set in \eqref{eqn:primal2}. Interestingly, our manipulations lead lower bounding $\tilM$ by~\eqref{eqn:sdp_opt}, which  is a (convex) semidefinite program \citep{Van96}.

Now observe that the approximations in~\eqref{eqn:dpq_approximations} and~\eqref{eqn:approximations} are accurate in the limit of large $n$ because the optimizing distribution $Q_n^*$ becomes increasingly close to $P_e$.  By continuity of the optimization problems in (perturbations of) the objective and the constraints, $\tilM(P_e;\veps_n)$ and $M(P_e;\veps_n)$ are close when $n$ is large, i.e.,
\begin{equation}
\lim_{n\to\infty} \left|\tilM(P_e; \veps_n)-M(P_e; \veps_n)\right|=0.\label{eqn:Mlimit}
\end{equation}
By applying the continuity statement above to \eqref{eqn:over_3}, for every $\eta >0$, there exists a $N\in\bN$ such that 
\begin{equation}
P^n( I (\hP_e) \ge \veps_n) \le (n+1)^{r^2}  \exp\left(-n      (\tilM(P_e;\veps_n)-\eta)    \right),  \nn
\end{equation}
for all $n>N$.  Define the constant
\begin{equation}
c_P:= \min_{e\in V\times V\,:\, I(P_e)=0 } \,\, \max_{\mu\ge 0} \,\, \{ \mu : \bPi_e  \succeq \mu\,  \bH_e  \}. \label{eqn:sdp}
\end{equation}
By~\eqref{eqn:sdp_opt}, \eqref{eqn:weak_duality} and the definition of $c_P$ in \eqref{eqn:sdp},
\begin{equation}
P^n( I (\hP_e) \ge \veps_n)  \le (n+1)^{r^2}  \exp\left(-n    \veps_n     (c_P-\eta)\right). \label{eqn:over_4}
\end{equation}
Putting \eqref{eqn:over_1}, \eqref{eqn:over_2} and \eqref{eqn:over_4} together, we see that the overestimation error  
\begin{equation}
P^n(\hk_n > k|\calB_n^c)\le(d-k-1)^2 (n+1)^{r^2}  \exp\left(-n \veps_n     (c_P-\eta) \right). \label{eqn:over_5}
\end{equation}
Note that the above probability tends to zero by the assumption that $n\veps_n/\log n\to \infty$ in~\eqref{eqn:veps}. Thus, we have consistency overall (since the underestimation, Chow-Liu and now the overestimation errors all tend to zero). Thus, by taking the normalized  logarithm (normalized by $n\veps_n$), the $\limsup$ in $n$   (keeping in mind that $d$ and $k$ are constant), we conclude that
\begin{equation}
\limsup_{n\to\infty} \frac{1}{n\veps_n}\log P^n(\hk_n > k|\calB_n^c)\le -(c_P-\eta). \label{eqn:over_error}
\end{equation}
If we now allow $\eta$ in \eqref{eqn:over_error} to tend to $0$, we see that it remains to prove that $c_P=1$ for all $P$. For this purpose, it suffices  to show that the optimal solution to the optimization problem in~\eqref{eqn:sdp_opt}, denoted  $\mu^*$, is equal to one for all $\bPi_e$ and $\bH_e$. 
Note that $\mu^*$ can be expressed in terms of eigenvalues:
\begin{equation}
\mu^* = \left(\max \left\{ \eig (\bPi_e^{-1/2} \bH_e \bPi_e^{-1/2} )  \right\} \right)^{-1}, \label{eqn:mustar}
\end{equation}
where $\eig(\bA)$ denotes the set of real eigenvalues of the symmetric matrix $\bA$. By using the definitions of $\bPi_e$ and $\bH_e$ in~\eqref{eqn:Pi_e} and~\eqref{eqn:H_e} respectively, we can verify that the matrix $\mathbf{I}-\bPi_e^{-1/2} \bH_e \bPi_e^{-1/2}$ is positive semidefinite with an eigenvalue at zero.  This proves that the largest eigenvalue of $\bPi_e^{-1/2} \bH_e \bPi_e^{-1/2}$ is one and hence from \eqref{eqn:mustar}, $\mu^*=1$. The proof of the upper bound in~\eqref{eqn:limsup_statement} is completed by combining the estimates in~\eqref{eqn:reduction}, \eqref{eqn:under_error} and \eqref{eqn:over_error}. 
\subsection*{Proof of lower bound in Theorem~\ref{thm:forest}}
The key idea is to bound the overestimation error using a modification of the lower bound in Sanov's theorem.  Denote the set of types supported on a finite set $\calY$ with denominator $n$ as $\calP_n(\calY)$ and the {\em type class} of a distribution $Q\in \calP_n(\calY)$ as $$\bbT(Q):=\{y^n\in\calY^n: \hP(\,\cdot\,; y^n) =Q(\,\cdot\,)\},$$ where $\hP(\,\cdot\,; y^n)$ is the empirical distribution of the sequence $y^n=(y_1,\ldots, y_n)$. The following bounds on the type class are well known \citep[Ch.\ 11]{Cov06}.
\begin{lemma}[Probability of Type Class] \label{lem:typeclass}
For any $Q\in\calP_n(\calY)$ and any distribution $P$, the probability of the type class $\bbT(Q)$ under $P^n$ satisfies:
\begin{align}
(n+1)^{-|\calY|} &\exp(-nD(Q\, ||\, P))\le P^n(\bbT(Q)) \le \exp(-nD(Q\, ||\, P)). \label{eqn:lower_bd_type} 
\end{align}
\end{lemma}
To prove the lower bound in \eqref{eqn:liminf_statement}, assume that $k<d-1$ and note that the error probability $P^n(\hk_n\ne k|\calB_n^c)$ can be lower bounded by $P^n(I(\hP_e) \ge \veps_n)$ for any node pair $e$ such that  $I(P_e)=0$. We seek to lower bound the latter probability by appealing to \eqref{eqn:lower_bd_type}. Now choose a sequence of distributions $Q^{(n)}\in \{Q \in\calP_n(\calX^2): I(Q)\ge\veps_n\}$ such that
\begin{equation}
\lim_{n\to\infty} \left| M(P_e;\veps_n)-D(Q^{(n)} \, || \, P_e) \right| = 0. \nn
\end{equation}
This is possible because the set of types is dense in the probability simplex \citep[Lemma 2.1.2(b)]{Dembo}.  Thus, 
\begin{align}
P^n(I(\hP_e) \ge \veps_n) &=\sum_{Q \in\calP_n(\calX^2):  I(Q) \ge \veps_n} P^n(\bbT(Q)) \nn  \\
&\ge P^n(\bbT(Q^{(n)}) )  \nn  \\
&\ge (n+1)^{-r^2} \exp(-nD(Q^{(n)}\,||\, P_e)), \label{eqn:method_types}
\end{align}
where~\eqref{eqn:method_types} follows from the lower bound in~\eqref{eqn:lower_bd_type}.  By applying~\eqref{eqn:Mlimit},  and using the fact that if $|a_n-b_n|\to 0$ and $|b_n-c_n|\to 0$ then, $|a_n-c_n|\to 0$ (triangle inequality), we also have  
\begin{equation}
\lim_{n\to\infty} \left| \tilM(P_e;\veps_n) - D(Q^{(n)} \, || \, P_e)\right| = 0. \nn
\end{equation}
Hence,  continuing the chain in \eqref{eqn:method_types},  for any $\eta>0$, there exists a $N\in\bN$ such that for all $n>N$,
 \begin{eqnarray}
P^n(I(\hP_e)  \ge  \veps_n)    \ge  (n+1)^{-r^2} \exp(-n (\tilM(P_e;\veps_n)+\eta)). \label{eqn:lower_2}
\end{eqnarray}
Note that an upper bound for $\tilM(P_e; \veps_n)$ in~\eqref{eqn:primal2} is simply given by the objective evaluated at any feasible point. In fact, by manipulating~\eqref{eqn:primal2}, we see that the upper bound is also proportional to $\veps_n$, i.e., 
\begin{equation}
\tilM(P_e; \veps_n)\le C_{P_e}\veps_n, \nn
\end{equation}
where $C_{P_e}\in (0,\infty)$ is some constant\footnote{We can easily remove the constraint $\bz^T \mathbf{1}$ in \eqref{eqn:primal2} by a simple change of variables to only consider those vectors in the subspace orthogonal to the all ones vector so we ignore it here for simplicity. To obtain $C_{P_e}$, suppose the matrix $\bW_e$ diagonalizes $\bH_e$, i.e., $\bH_e  = \bW_e^T\bD_e\bW_e$, then one can, for example, choose $C_{P_e}  =   \min_{i:[\bD_e]_{i,i}>0 } [\bW_e^T  \bPi_e \bW_e]_{i,i}. $}  that depends on the matrices  $\bPi_e$ and $\bH_e$. Define $C_P := \max_{e\in V\times V:I(P_e)=0} C_{P_e}$. Continuing the lower bound in~\eqref{eqn:lower_2}, we obtain
 \begin{align}
P^n(I(\hP_e)  \ge  \veps_n)   
\ge  (n+1)^{-r^2} \exp(-n\veps_n (C_{P}+\eta) ), \nn
\end{align}
for $n$  sufficiently large. Now take the normalized logarithm and the $\liminf$ to conclude that
\begin{equation}
\liminf_{n\to\infty} \frac{1}{n\veps_n}\log P^n(\hk_n\ne k|\calB_n^c)   \ge  -  (C_P+\eta).  \label{eqn:structure_lb}
\end{equation}
Substitute~\eqref{eqn:structure_lb} into the lower bound in~\eqref{eqn:reduction_lower}. Now the resulting inequality is independent of $n$ and we can take $\eta\to 0$ to complete the  proof of the lower bound in Theorem~\ref{thm:forest}.  
\subsection*{Proof of the exponential rate of decay for trees in Theorem~\ref{thm:forest}}
For the claim in \eqref{eqn:exponentially_fast}, note that for $n$ sufficiently large, 
\begin{equation}
P^n(\calA_n) \ge \max\{ (1-\eta) P^n(\hk_n\ne k_n|\calB_n^c) , P^n(\calB_n)\}, \label{eqn:error_exp_tree}
\end{equation}
from Lemma~\ref{lem:reduction} and the fact that $\calB_n\subseteq \calA_n$. Eq.~\eqref{eqn:error_exp_tree} gives us a lower bound on the error probability in terms of the Chow-Liu error $P^n(\calB_n)$ and the underestimation and overestimation errors $ P^n(\hk_n\ne k_n|\calB_n^c)$. If $k=d-1$, the overestimation error probability is identically zero, so we only have to be concerned with the underestimation error. Furthermore, from~\eqref{eqn:under_error} and a corresponding lower bound which we omit, the underestimation error event satisfies $P^n(\hk_n < k|\calB_n^c) \doteq \exp(-n L_P)$. Combining this fact with the definition of the error exponent $K_P$ in~\eqref{eqn:KPdef} and the result in~\eqref{eqn:error_exp_tree} establishes~\eqref{eqn:exponentially_fast}. Note that the relation in~\eqref{eqn:error_exp_tree} and our preceding upper bounds ensure that the limit in~\eqref{eqn:exponentially_fast} exists. \end{proof}

\begin{proof} ({\em of Lemma~\ref{lem:reduction}})
We note that  $P^n(\calA_n| \hk_n\ne k)=1$ and thus, 
\begin{align}
P^n(\calA_n)   \leq P^n (\hk_n\ne k) + P^n(\calA_n| \hk_n =  k).\label{eqn:upperbound1}
\end{align}
By using the definition of $K_P$ in \eqref{eqn:KPdef}, the second term in \eqref{eqn:upperbound1} is precisely $P^n(\calB_n)$ therefore,
\begin{equation}
P^n(\calA_n) \le P^n (\hk_n\ne k) +\exp(-n (K_P-\eta)), \label{eqn:upperbound2}
\end{equation}
for all $n>N_1$. We further bound $P^n (\hk_n\ne k)$ by conditioning on the event $\calB_n^c$. Thus, for $\eta>0$,
\begin{align}
P^n (\hk_n\ne k)  &\le  P^n (\hk_n\ne k|\calB_n^c) + P^n( \calB_n)\nn\\
&\le P^n (\hk_n\ne k|\calB_n^c)+ \exp(-n (K_P-\eta)),  \label{eqn:upperbound3}
\end{align}
for all $n>N_2$. The upper bound result follows by combining~\eqref{eqn:upperbound2} and~\eqref{eqn:upperbound3}. The lower bound follows by the chain 
\begin{align}
P^n(\calA_n) &\ge P^n (\hk_n\ne k)\ge P^n (\{\hk_n\ne k\}\cap\calB_n^c)  \nn\\
&=P^n (\hk_n\ne k|\calB_n^c)P^n(\calB_n^c)  \ge (1-\eta)P^n (\hk_n\ne k|\calB_n^c), \nn
\end{align}
which holds for all $n>N_3$ since $P^n(\calB_n^c)\to 1$. Now the claims in \eqref{eqn:reduction_lower} and \eqref{eqn:reduction}  follow  by taking $N  :=\max\{N_1,N_2,N_3\}$.  
\end{proof}

\section{Proof of Corollary~\ref{cor:forest_proj}} \label{prf:cor:forest_proj}
\begin{proof} 
This claim follows from the fact that three errors (i) Chow-Liu error (ii) underestimation error and (iii) overestimation error behave in exactly the same way as in Theorem~\ref{thm:forest}. In particular, the Chow-Liu error, i.e., the error for the learning the top $k$ edges in the forest projection model $\tilP$ decays with error exponent $K_P$. The underestimation error behaves as in \eqref{eqn:under_error} and the overestimation error as in \eqref{eqn:over_error}.
\end{proof}

\section{Proof of Theorem~\ref{cor:scaling}} \label{prf:cor:scaling}
\begin{proof}
Given assumptions (A1) and (A2), we claim that the underestimation exponent $L_{P^{(d)}}$, defined in~\eqref{eqn:under_error}, is  uniformly bounded away from zero, i.e., 
\begin{equation}
L:=\inf_{d\in\bN}L_{P^{(d)}}=\inf_{d\in\bN}\min_{e \in E_{P^{(d)}}} L( P^{(d)}_{e};0) \label{eqn:Lconst}
\end{equation}
is positive. Before providing a formal proof, we provide a plausible argument to show that this claim is true.  Recall the definition of $L(P_{e};0)$  in~\eqref{eqn:LPe}. Assuming that the joint $P_{e}=P_{i,j}$ is close to a product distribution or equivalently if its mutual information $I(P_e)$ is small (which is the worst-case scenario), 
\begin{align}
L(P_{e};0) &\approx \min_{Q\in\calP(\calX^2)}\{D(P_{e}\,||\,Q):I(Q)=0\} \label{eqn:swap_args} \\
 &=D(P_{e} \, ||\, P_i \, P_j)   =I(P_{e})\ge  I_{\inf}>0, \label{eqn:def_Imin}
\end{align}
where in \eqref{eqn:swap_args}, the arguments in the KL-divergence have been swapped. This is because when $Q\approx P_{e}$ entry-wise, $D(Q\, ||\, P_{e})\approx D(P_{e}\, ||\, Q)$ in the sense that their difference is small compared to their absolute values \citep{Bor08}.   In \eqref{eqn:def_Imin}, we used the fact that the reverse I-projection of $P_e$ onto the set of  product distributions is $P_i P_j$. Since  $I_{\inf}$ is constant, this proves the claim, i.e., $L>0$.  

More formally, let $$B_{\kappa'}:=\{Q_{i,j}\in\calP(\calX^2): Q_{i,j}(x_i,x_j)\ge \kappa',\forall\, x_i,x_j\in\calX\}$$ be the set of joint distributions whose entries are  bounded away from zero by $\kappa'>0$. Now, consider a pair of joint distributions $P_{e}^{(d)},\tilP_{e}^{(d)}\in B_{\kappa'}$ whose minimum values are uniformly bounded away from zero as assumed  in (A2). Then there exists a Lipschitz  constant (independent of $d$) $U\in (0,\infty)$ such that for all $d$,
\begin{equation}
|I(P_{e}^{(d)})- I(\tilP_{e}^{(d)})|\le  U \|\vect(P_{e}^{(d)})-\vect(\tilP_{e}^{(d)})\|_1,  \label{eqn:equi_lip}
\end{equation}
where $\|\cdot\|_1$ is the vector $\ell_1$ norm. In fact, $U:=\max_{Q\in B_{\kappa'}} \|\nabla I(\vect(Q))\|_{\infty}$ is the Lipschitz constant of $I(\cdot)$ which is uniformly bounded because the joint distributions $P_{e}^{(d)}$ and $\tilP_{e}^{(d)}$  are assumed to be uniformly bounded away from zero. Suppose, to the contrary, $L=0$. Then by the definition of the infimum in~\eqref{eqn:Lconst}, for every $\epsilon>0$, there exists a $d\in\bN$ and a corresponding $e\in E_{P^{(d)}}$  such that if   $Q^*$ is the optimizer in~\eqref{eqn:LPe},
\begin{equation}
\epsilon> D(Q^* \, || \, P_{e}^{(d)} ) \stackrel{(a)}{\ge} \frac{\|\vect(P_{e}^{(d)})-\vect(Q^*)\|_1^2}{2\log 2}    
\stackrel{(b)}{\ge} \frac{|I(P_{e}^{(d)})-I(Q^*)|^2 }{(2\log 2) U^2}    
\stackrel{(c)}{\ge} \frac{I_{\inf}^2}{(2\log 2) U^2} ,  \nn 
\end{equation}
where (a) follows from  Pinsker's inequality \cite[Lemma 11.6.1]{Cov06}, (b) is an application of~\eqref{eqn:equi_lip} and the fact that if $P_{e}^{(d)}\in B_{\kappa}$ is uniformly bounded from zero (as assumed in \eqref{eqn:minentry}) so is the associated optimizer $Q^*$ (i.e., in $B_{\kappa''}$ for some possibly different uniform $\kappa''>0$). Statement (c) follows from the definition of $I_{\inf}$ and the fact that $Q^*$ is a product distribution, i.e., $I(Q^*)=0$. Since $\epsilon$ can be chosen to be arbitrarily small, we arrive at a contradiction. Thus $L$ in \eqref{eqn:Lconst} is positive. Finally, we observe from~\eqref{eqn:under_4} that if $n>(3/L)\log k$ the underestimation error tends to zero because \eqref{eqn:under_4}  can be further upper bounded as 
\begin{align}
P^n(\hk_n < k|\calB_n^c) &\le(n+1)^{r^2}\exp(2\log k-nL)  
< (n+1)^{r^2}\exp\left(\frac{2}{3}nL-nL\right)\to 0 \nn
\end{align}
as $n\to \infty$. Take $C_2=3/L$ in \eqref{eqn:scaling}.

Similarly, given the same assumptions, the error exponent for structure learning  $K_{P^{(d)}}$, defined in~\eqref{eqn:KPdef}, is also uniformly bounded away from zero, i.e.,
\begin{equation}
K:=\inf_{d\in \bN} K_{P^{(d)}}>0. \nn
\end{equation}
Thus,   if $n>(4/K)\log d$,  the error probability associated to estimating  the top $k$ edges (event $\calB_n$)  decays to zero along similar lines as in the case of the underestimation error. Take $C_1=4/K$  in~\eqref{eqn:scaling}.

Finally, from~\eqref{eqn:over_5}, if $n\veps_n>2\log(d-k)$, then the overestimation error tends to zero. Since from \eqref{eqn:veps}, $\veps_n$ can take the form $n^{-\beta}$ for $\beta>0$, this is equivalent to $n^{1-\beta}>2\log(d-k)$, which is the same as the first condition  in~\eqref{eqn:scaling}, namely $n>(2\log (d-k))^{1+\zeta}$. By~\eqref{eqn:reduction} and~\eqref{eqn:over_under}, these three probabilities constitute the overall error probability when learning the sequence of forest structures $\{E_{P^{(d)}}\}_{d\in\bN}$. Thus the conditions in~\eqref{eqn:scaling} suffice for high-dimensional consistency. 
\end{proof}

\section{Proof of Corollary~\ref{cor:extremal}} \label{prf:cor:extremal}
\begin{proof}
First note that $k_n\in\{ 0,\ldots, d_n-1\}$. From~\eqref{eqn:over_5},  we see that for $n$ sufficiently large, the sequence $h_n(P):=(n\veps_n)^{-1} \log P^n(\calA_n)$ is upper bounded by 
\begin{align}
- 1 +\frac{2}{n\veps_n}\log (d_n-k_n-1)    + \frac{r^2\log(n+1)}{n\veps_n}. \label{eqn:normalized_log}
\end{align}
The last term in \eqref{eqn:normalized_log}  tends to zero by~\eqref{eqn:veps}. Thus $h_n(P)=O((n\veps_n)^{-1}\log (d_n-k_n-1))$, where the implied constant is 2 by \eqref{eqn:normalized_log}. Clearly, this sequence is    maximized (resp.\ minimized) when $k_n=0$ (resp.\ $k_n=d_n-1$).   Eq.~\eqref{eqn:normalized_log} also shows that the sequence $h_n$ is monotonically decreasing in $k_n$. 
\end{proof}

\section{Proof of Theorem~\ref{thm:sample}} \label{prf:thm:sample}
\begin{proof}
We first focus on part (a). Part (b) follows in a relatively straightforward manner. Define
\begin{equation}
\hT_{\MAP}(\bx^n):=\argmax_{t\in \calT_k^d}\,\, \bP(T_P=t|\bx^n) \nn
\end{equation}
to be the maximum a-posteriori (MAP) decoding rule.\footnote{In fact, this proof works  for {\em any} decoding rule, and not just the MAP rule. We focus on the MAP rule for concreteness. } By the optimality of the MAP rule, this lower bounds the error probability of 
any other estimator. Let $\calW:=\hT_{\MAP}((\calX^d)^n)$ be the range of the function $\hT_{\MAP}$, i.e., a forest $t\in\calW$ if and only if there exists a sequence $\bx^n$ such that $\hT_{\MAP}=t$. Note that $\calW\cup \calW^c = \calT_k^d$. Then, consider the  lower bounds:
\begin{align}
\bP(\hT \ne T_P ) &= \sum_{t\in \calT_k^d} \bP(\hT \ne T_P|T_P=t) \bP(T_P=t) \nn \\*
&\ge \sum_{t\in \calW^c} \bP(\hT \ne T_P|T_P=t) \bP(T_P=t) \nn \\*
&= \sum_{t\in \calW^c} \bP(T_P=t)\label{eqn:prob_one}= 1- \sum_{t\in \calW} \bP(T_P=t)\\
&= 1- \sum_{t\in \calW} |\calT_k^d|^{-1} \label{eqn:uniformity} \\
&\ge  1 - r^{nd} |\calT_k^d|^{-1}, \label{eqn:cardinality_bound}
\end{align}
where in~\eqref{eqn:prob_one}, we used the fact that $\bP(\hT \ne T_P|T_P=t) =1$ if $t\in \calW^c$, in \eqref{eqn:uniformity}, the fact that $\bP(T_P=t)=1/|\calT_k^d|$. In \eqref{eqn:cardinality_bound}, we used the observation  $|\calW|\le (|\calX^d|)^n=r^{nd}$ since the function $\hT_{\MAP}:(\calX^d)^n \to \calW$ is surjective. Now, the number of labeled forests with $k$ edges and  $d$ nodes is \citep[pp. 204]{Aigner09}
$
|\calT_k^d|\ge (d-k)d^{k-1}\ge d^{k-1}.
$
Applying this lower bound to~\eqref{eqn:cardinality_bound}, we obtain
\begin{align}
\bP(\hT \ne T_P ) \ge  1-\exp\left(nd \log r- (k-1) \log d\right)> 1-\exp\left((\varrho-1) (k-1)\log d  \right), \label{eqn:choice_n}
\end{align}
where the second inequality follows by  choice of $n$ in \eqref{eqn:sample_lower_bd}.  
The estimate in~\eqref{eqn:choice_n} converges to 1  as $(k,d)\to\infty$ since $\varrho<1$. The same reasoning applies to part (b) but we instead  use the following estimates of the cardinality of the set of forests \cite[Ch.\ 30]{Aigner09}:
\begin{equation}
(d-2)\log d\le\log|\calF^d| \le (d-1)\log (d+1). \label{eqn:num_forests}
\end{equation}
Note that we have lower bounded $|\calF^d|$ by the number trees with $d$ nodes which is $d^{d-2}$ by Cayley's formula \citep[Ch.\ 30]{Aigner09}. The upper bound\footnote{The purpose of the upper bound is to show that our estimates of $|\calF^d|$ in~\eqref{eqn:num_forests} are reasonably tight.} follows by a simple combinatorial argument which is omitted. Using the lower bound in~\eqref{eqn:num_forests}, we have
\begin{align}
\bP(\hT\ne T_P) \ge 1-\exp(nd\log r) \exp( -(d-2)\log d)  > 1-d^2\exp ((\varrho-1)d\log d), \label{eqn:estimate_all_forests}
\end{align}
with the choice of $n$ in \eqref{eqn:sample_lower_bd_all_forests}. The estimate in  \eqref{eqn:estimate_all_forests} converges to 1, completing the proof. 
\end{proof}

\section{Proof of Theorem \ref{thm:consisten_param}} \label{prf:thm:consisten_param}
\begin{proof}
We assume that $P$ is Markov on a forest since the extension to  non-forest-structured $P$ is a straightforward generalization. We start with some useful definitions.  Recall from Appendix~\ref{prf:thm:forest} that $\calB_n:=\{\hE_k\ne E_P\}$ is the event that the top $k$ edges (in terms of mutual information) in the edge set $\hE_{d-1}$ are not equal to the edges in $E_P$. Also define $\tilcalC_{n,\delta}:= \{D(P^* \, ||\, P)>\delta d\}$ to be the event that the divergence between the learned model and the true (forest) one is greater than $\delta d$. We will see that $\tilcalC_{n,\delta}$ is closely related to the event of interest $\calC_{n,\delta}$ defined in \eqref{eqn:Cn}. Let $\calU_n:=\{\hk_n<k\}$ be the underestimation event. Our proof relies on the following result, which is similar to Lemma~\ref{lem:reduction}, hence its proof is omitted.

\begin{lemma} \label{lem:reduction_2}
For every $\eta>0$, there exists a $N \in\bN$ such that for all $n>N$, the following bounds on $ P^n(\tilcalC_{n,\delta})$ hold:
\begin{align}
&(1-\eta)P^n(\tilcalC_{n,\delta} |\calB_n^c,\calU_n^c)\le P^n(\tilcalC_{n,\delta})  \label{eqn:lowerboundDPstar}\\*
&\le P^n(\tilcalC_{n,\delta}|\calB_n^c,\calU_n^c)+\exp(-n(\min\{K_P,L_P\}-\eta)). \label{eqn:boundDPstar}
\end{align}
\end{lemma}
Note that the exponential term in \eqref{eqn:boundDPstar} comes from an application of the union bound and the ``largest-exponent-wins'' principle in large-deviations theory.  From \eqref{eqn:lowerboundDPstar} and \eqref{eqn:boundDPstar} we see that it is possible to bound the probability of $\tilcalC_{n,\delta}$ by providing upper  and lower  bounds for  $P^n(\tilcalC_{n,\delta}|\calB_n^c,\calU_n^c)$. In particular, we  show that the upper bound equals $\exp(-n\delta)$ to first order in the exponent. This will lead directly to \eqref{eqn:param_upper}. To proceed, we rely on the following lemma, which is a generalization of a well-known result~\cite[Ch.\ 11]{Cov06}. We defer the proof to the end of the section. 
\begin{lemma}[Empirical Divergence  Bounds]  \label{lem:dpq_bound}
Let $X,Y$ be two  random variables whose joint distribution is $P_{X,Y}\in\calP(\calX^2)$ and $|\calX|=r$. Let $(x^n,y^n)=\{(x_1,y_1),\ldots, (x_n,y_n)\}$ be $n$ independent and identically distributed observations drawn from $P_{X,Y}$. Then, for every $n$,
\begin{equation}
P_{X,Y}^n(D(\hP_{X|Y}\, ||\, P_{X|Y})>\delta)\le (n+1)^{r^2} \exp(-n\delta), \label{eqn:kl_div_bound} 
\end{equation}
where $\hP_{X|Y}=\hP_{X,Y}/\hP_Y$ is the conditional type of $(x^n, y^n)$. Furthermore, 
\begin{equation}
\liminf_{n\to\infty} \frac{1}{n}\log P_{X,Y}^n(D(\hP_{X|Y}\, ||\, P_{X|Y})>\delta)\ge -\delta. \label{eqn:kl_div_lower_bound} 
\end{equation}
\end{lemma}
\begin{figure}
\centering
\begin{tabular}{cc}
\begin{picture}(100,50)
\linethickness{0.22mm}
\put(0,0){\circle*{8}}
\put(50,0){\circle*{8}}
\put(100,0){\circle*{8}}
\color{blue}
\put(0,50){\circle*{8}}
\color{black}
\put(50,50){\circle*{8}}
\put(-12,50){\mbox{ 1}}
\put(-12,4){\mbox{ 2}}
\put(37,50){\mbox{ 4}}
\put(37,04){\mbox{ 3}}
\put(87,50){\mbox{ 5}}
\put(87,04){\mbox{ 6}}
\put(0,50){\vector(0,-1){47}}
\put(0,0){\vector(1,0){47}}
\put(50,0){\vector(0,1){47}}
\put(100,50){\vector(0,-1){47}}
\put(67,20){ \mbox{$\hE_{\hk_n}$}}
\color{blue}
\put(100,50){\circle*{8}}
\color{black}
\end{picture}
\hspace{1in}
\begin{picture}(100,50)
\linethickness{0.22mm}
\put(0,0){\circle*{8}}
\put(50,0){\circle*{8}}
\color{blue}
\put(100,0){\circle*{8}}
\color{blue}
\put(0,50){\circle*{8}}
\color{black}
\put(50,50){\circle*{8}}
\put(-12,50){\mbox{ 1}}
\put(-12,4){\mbox{ 2}}
\put(37,50){\mbox{ 4}}
\put(37,04){\mbox{ 3}}
\put(87,50){\mbox{ 5}}
\put(87,04){\mbox{ 6}}
\put(0,50){\vector(0,-1){47}}
\put(0,0){\vector(1,0){47}}
\put(50,0){\vector(0,1){47}}
\put(67,20){ \mbox{$E_P$}}
\color{blue}
\put(100,50){\circle*{8}}
\color{black}
\end{picture}
\end{tabular}
\caption{In $\hE_{\hk_n}$ (left), nodes 1 and 5 are the roots, which are in blue. The parents are defined as  $\pi(i;\hE_{\hk_n}) =i-1$ for $i = 2,3,4,6$ and $\pi(i;\hE_{\hk_n}) =\emptyset$ for $i=1,5$. In $E_P$ (right), the parents are defined as $\pi(i;E_P )=i-1$ for $i = 2,3,4$ but $\pi(i;E_P )=\emptyset$ for $i=1,5,6$ since $(5,6),(\emptyset,1),(\emptyset,5) \notin E_P$. } \label{fig:directed_eg}
\end{figure}
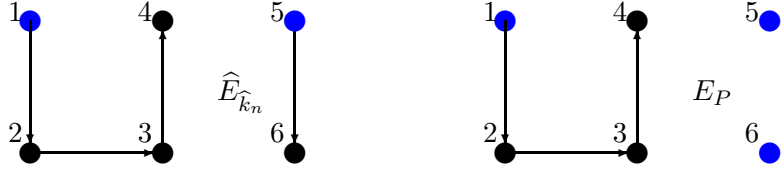

It is worth noting that the  bounds in~\eqref{eqn:kl_div_bound}  and~\eqref{eqn:kl_div_lower_bound}  are independent of the distribution $P_{X,Y}$ (cf.\ discussion after Theorem~\ref{thm:consisten_param}). We now proceed with the proof of Theorem~\ref{thm:consisten_param}. To do so, we consider the directed representation of a tree distribution $Q$ \citep{Lau96}:
\begin{equation}
Q(\bx) =\prod_{i\in V} Q_{i|\pi(i)} (x_i|x_{\pi(i)}), \label{eqn:tree_decomp_directed} 
\end{equation}
where $\pi(i)$ is the parent of $i$ in the edge set of $Q$ (assuming a fixed root). Using \eqref{eqn:tree_decomp_directed}  and conditioned on the fact that the top $k$ edges of the graph of $P^*$ are the same as those in $E_P$ (event $\calB_n^c$) and underestimation does not occur (event $\calU_n^c$), the KL-divergence between $P^*$ (which is a function of the samples $\bx^n$ and hence of $n$) and $P$ can be expressed as a sum over $d$ terms:
\begin{equation}
D(P^* \, ||\, P)=\sum_{i\in V}D(\hP_{i|\pi(i;\hE_{\hk_n})}\, ||\, P_{i|\pi(i;E_P)}), \label{eqn:divergence_decomp}
\end{equation}
where the parent of node $i$ in $\hE_{\hk_n}$, denoted $\pi(i;\hE_{\hk_n})$, is defined by arbitrarily choosing a root in each component tree of the forest $\hT_{\hk_n}=(V,\hE_{\hk_n})$. The parents of the chosen roots are empty sets.  The parent of node $i$ in $E_P$ are ``matched'' to those in $\hE_{\hk_n}$, i.e.,  defined  as $\pi(i;E_P):=\pi(i;\hE_{\hk_n})$  if $(i,\pi(i;\hE_{\hk_n}))\in E_P$ and $\pi(i;E_P):=\emptyset$ otherwise. See Figure~\ref{fig:directed_eg} for an example. Note that this can be done because $\hE_{\hk_n}\supseteq E_P$ by conditioning on the events $\calB_n^c $ and $\calU_n^c=\{\hk_n\ge k\}$. Then, the error probability $P^n(\tilcalC_{n,\delta}|\calB_n^c,\calU_n^c)$ in~\eqref{eqn:boundDPstar} can be upper bounded as 
\begin{align}
P^n(\tilcalC_{n,\delta}|\calB_n^c,\calU_n^c)&=  P^n\left(\sum_{i\in V}D(\hP_{i|\pi(i;\hE_{\hk_n})}|| P_{i|\pi(i;E_P)})>\delta d \Big|\calB_n^c,\calU_n^c\right) \label{eqn:parents_decomp}\\
&=  P^n \left( \frac{1}{d}\sum_{i\in V}D(\hP_{i|\pi(i;\hE_{\hk_n})}|| P_{i|\pi(i;E_P)}) > \delta \Big|\calB_n^c,\calU_n^c \right) \nn\\
&\le  P^n \left( \max_{i\in V}\,  \left\{  D(\hP_{i|\pi(i;\hE_{\hk_n})} || P_{i|\pi(i;E_P)}) \right\}  > \delta  \Big|\calB_n^c,\calU_n^c \right) \label{eqn:max_bound}\\
& \le  \sum_{i\in V} P^n\left( D(\hP_{i|\pi(i;\hE_{\hk_n})}|| P_{i|\pi(i;E_P)})   > \delta  \Big|\calB_n^c ,\calU_n^c \right) \label{eqn:union_bound_max}\\
& \le  \sum_{i\in V}  (n + 1)^{r^2}   \exp\left(- n \delta \right)   =  d  (n  +  1)^{r^2}  \exp  \left(-n \delta \right), \label{eqn:abovelemma}
\end{align}
where Eq.~\eqref{eqn:parents_decomp} follows from the decomposition in~\eqref{eqn:divergence_decomp}.  Eq.~\eqref{eqn:max_bound} follows from the fact that if the arithmetic mean of $d$ positive numbers exceeds $\delta$, then the maximum exceeds $\delta$. Eq.~\eqref{eqn:union_bound_max} follows from the union bound. Eq.~\eqref{eqn:abovelemma}, which holds for all $n\in\bN$, follows from the upper bound in~\eqref{eqn:kl_div_bound}. Combining~\eqref{eqn:boundDPstar} and~\eqref{eqn:abovelemma} shows that if $\delta<\min\{K_P, L_P\}$,
\begin{equation}
\limsup_{n\to\infty}\frac{1}{n}\log P^n(\tilcalC_{n,\delta}) \le  - \delta. \nn 
\end{equation}
Now recall that $\tilcalC_{n,\delta}=\{D(P^*\, ||\, P)>\delta d\}$.  In order to complete the proof of \eqref{eqn:param_upper}, we need to swap the arguments in the KL-divergence to bound the probability of the event $\calC_{n,\delta} = \{D(P\, ||\, P^*)>\delta d\}$. To this end, note that for   $\epsilon>0$ and $n$ sufficiently large, $|D(P^* \, ||\, P)-D(P \, ||\, P^*)|<\epsilon$ with high probability  since the two KL-divergences become close ($P^*\approx P$  w.h.p.\  as $n\to\infty$). More precisely, the probability of   $\{|D(P^* \, ||\, P)-D(P \, ||\, P^*)|\ge\epsilon\}=\{o(\|P-P^*\|_{\infty}^2)\ge \epsilon\}$   decays exponentially with some   rate  $M_P>0$. Hence,
\begin{equation}
\limsup_{n\to\infty}\frac{1}{n}\log P^n( D( P\,||\, P^*) >\delta d) \le  - \delta,\label{eqn:param_intermediate}
\end{equation}
if $\delta<\min\{K_P, L_P,M_P\}$. If $P$ is not Markov on a forest,  \eqref{eqn:param_intermediate} holds with the forest projection $\tilP$ in place of $P$, i.e.,
\begin{equation}
\limsup_{n\to\infty}\frac{1}{n}\log P^n( D( \tilP\,||\, P^*) >\delta d) \le  - \delta.\label{eqn:param_intermediate2}
\end{equation}
The  Pythagorean relationship \citep{Simon73, Bach03} states that
\begin{equation}
D(P \, ||\, P^*) = D(P\, ||\, \tilP)  + D(\tilP \,|| \, P^*)  \label{eqn:Pyth2}
\end{equation}
which means that the risk is $\calR_n(P^*)= D(\tilP \,|| \, P^*)$. Combining this fact with \eqref{eqn:param_intermediate2} implies the assertion of~\eqref{eqn:param_upper} by choosing $\delta_0:=\min\{K_P, L_P, M_P\}$.

Now we exploit the lower bound in Lemma~\ref{lem:dpq_bound} to prove the lower bound in Theorem~\ref{thm:consisten_param}. The error probability $P^n(\tilcalC_{n,\delta}|\calB_n^c,\calU_n^c)$ in~\eqref{eqn:boundDPstar} can now be lower bounded by
\begin{align}
P^n(\tilcalC_{n,\delta}|\calB_n^c,\calU_n^c)& \ge \max_{i\in V}\,\, P^n\left(  D(\hP_{i|\pi(i;\hE_{\hk_n})}\, ||\, P_{i|\pi(i;E_P)})  >\delta d\Big|\calB_n^c,\calU_n^c\right) \label{eqn:lower_d} \\
&\ge  \exp(-n(\delta d+\eta)), \label{eqn:lower_d_exponent}
\end{align}
where~\eqref{eqn:lower_d} follows from the decomposition in \eqref{eqn:parents_decomp} and~\eqref{eqn:lower_d_exponent} holds for every $\eta$ for sufficiently large $n$ by \eqref{eqn:kl_div_lower_bound}. Using the same argument that allows us to swap the arguments of the KL-divergence as in the proof of the upper bound completes the proof of~\eqref{eqn:param_lower}.
\end{proof}


\begin{proof} ({\em of Lemma~\ref{lem:dpq_bound}})
Define the {\em $\delta$-conditional-typical set with respect to $P_{X,Y}\in\calP(\calX^2)$} as
\begin{equation}
\calS_{P_{X,Y}}^{\delta}:= \{ (x^n,y^n)\in (\calX^2)^n:D(\hP_{X|Y}\, ||\, P_{X|Y})\le\delta\}, \nn
\end{equation}
where $\hP_{X|Y}$ is the conditional type of $(x^n,y^n)$. We now estimate the $P_{X,Y}^n$-probability of the $\delta$-conditional-atypical set, i.e., $P^n_{X,Y}((\calS_{P_{X,Y}}^{\delta})^c) $
\begin{align}
P^n_{X,Y}((\calS_{P_{X,Y}}^{\delta})^c)&= \sum_{(x^n,y^n)\in\calX^2 :D(\hat{P}_{X|Y}||P_{X|Y})>\delta} P^n_{X,Y}((x^n, y^n))  \label{eqn:sum_sequences}\\
&= \sum_{Q_{X,Y}\in\calP_n(\calX^2) :D(Q_{X|Y}||P_{X|Y})>\delta} P^n_{X,Y}(\bbT(Q_{X,Y}))  \label{eqn:sum_type_class}\\
&\le \sum_{Q_{X,Y} \in\calP_n(\calX^2) :D(Q_{X|Y}||P_{X|Y})>\delta} \exp(-nD(Q_{X,Y}\, ||\, P_{X,Y})) \label{eqn:method_types_upper}\\
&\le \sum_{Q_{X,Y}\in\calP_n(\calX^2) :D(Q_{X|Y}||P_{X|Y})>\delta}   \exp(-n D(Q_{X|Y}\, ||\, P_{X|Y})) \label{eqn:div_pos} \\
&\le \sum_{Q_{X,Y}\in\calP_n(\calX^2)  :D(Q_{X|Y}||P_{X|Y})>\delta}  \exp(-n \delta) \label{eqn:delta_bound} \\
&\le (n+1)^{r^2} \exp(-n \delta  ) ,\label{eqn:num_types}
\end{align}
where~\eqref{eqn:sum_sequences} and \eqref{eqn:sum_type_class} are the same because summing over sequences is equivalent to summing over the corresponding type classes since every sequence in each type class has the same probability \citep[Ch.\ 11]{Cov06}. Eq.~\eqref{eqn:method_types_upper} follows from the method of types result in Lemma~\ref{lem:typeclass}. Eq.~\eqref{eqn:div_pos} follows from the KL-divergence version of the chain rule,  namely, 
$$
D(Q_{X,Y}\, ||\, P_{X,Y})=D(Q_{X|Y}\, ||\, P_{X|Y}) +D(Q_Y\, ||\, P_Y)
$$
and non-negativity of the KL-divergence $D(Q_Y\, ||\, P_Y)$. Eq.~\eqref{eqn:delta_bound} follows from the fact that $D(Q_{X|Y}\, ||\, P_{X|Y})>\delta$ for $Q_{X,Y}\in(\calS_{P_{X,Y}}^{\delta})^c$. Finally,~\eqref{eqn:num_types} follows the fact that the number of types with denominator $n$ and alphabet $\calX^2$ is upper bounded by  $(n+1)^{r^2}$. This concludes the proof of~\eqref{eqn:kl_div_bound}.

We now prove the lower bound in~\eqref{eqn:kl_div_lower_bound}. To this end, construct a sequence of distributions  $\{Q_{X,Y}^{(n)}\in\calP_n(\calX^2)\}_{n\in\bN}$ such that $Q_{Y}^{(n)}=P_{Y}$ and $D(Q_{X|Y}^{(n)}||P_{X|Y})\to\delta$. Such a sequence exists by the denseness of types in the probability simplex \citep[Lemma 2.1.2(b)]{Dembo}. Now we lower bound~\eqref{eqn:sum_type_class}:
\begin{align}
P^n_{X,Y}(  &(\calS_{P_{X,Y}}^{\delta})^c) \ge P_{X,Y}^n  ( \bbT(Q_{X,Y}^{(n)}))\ge (n+1)^{-r^2}\exp(-nD(Q_{X,Y}^{(n)}\, ||\, P_{X,Y})). \label{eqn:lower_cond_type}
\end{align}
Taking the normalized logarithm  and  $\liminf$ in $n$ on both sides of \eqref{eqn:lower_cond_type} yields
\begin{align}
 \liminf_{n\to\infty} \frac{1}{n}\log P^n_{X,Y}(  (\calS_{P_{X,Y}}^{\delta})^c)   \ge \liminf_{n\to\infty}  \left\{  -D(Q_{X|Y}^{(n)}|| P_{X|Y}) -  D(Q_{Y}^{(n)}|| P_{Y})  \right\} = -\delta. \nn
\end{align} 
This concludes the proof of Lemma~\ref{lem:dpq_bound}.
\end{proof}


\section{Proof of Corollary~\ref{cor:scaling_param}} \label{prf:cor:scaling_param}
\begin{proof}
If the dimension $d=o(\exp(n\delta))$, then the  upper bound in \eqref{eqn:abovelemma}  is asymptotically majorized by $\poly(n) o(\exp(na))\exp(-n\delta)= o(\exp(n\delta))\exp(-n\delta)$, which can be made arbitrarily small for $n$ sufficiently large.  Thus the probability tends to zero as $n\to\infty$. 
\end{proof}

\section{Proof of Theorem~\ref{cor:estimation_consist}} \label{prf:cor:estimation_consist}
\begin{proof}
In this proof, we drop the superscript $(d)$ for all distributions $P$ for notational simplicity but note that $d=d_n$. We first claim that  $D(P^*\, ||\, \tilP)=O_p(d\log d/n^{1-\gamma})$. Note from~\eqref{eqn:boundDPstar} and~\eqref{eqn:abovelemma} that by taking $\delta =(\tau \log d)/n^{1-\gamma}$ (for any $\tau>0$),
\begin{equation}
P^n\left(\frac{n^{1-\gamma}}{d\log d } \, D(P^*\, ||\, \tilP)>\tau\right)\le d(n+1)^{r^2} \exp(-\tau n^{\gamma} \log d ) +\exp(-\Theta(n))=o_n(1). \label{eqn:final_param}
\end{equation} 
Therefore,  the scaled sequence of random variables $\frac{n^{1-\gamma}}{d\log d}D(P^*\, ||\, \tilP)$ is stochastically bounded \citep{Serfling} which proves the claim.\footnote{In fact,   we have in fact proven the stronger assertion that $D(P^*\, ||\, \tilP)=o_p(d\log d/n^{1-\gamma})$ since the right-hand-side of \eqref{eqn:final_param} converges to zero.}

Now, we claim that $D(\tilP\, ||\, P^*) =  O_p(d\log d/n^{1-\gamma})$. A simple calculation  using Pinsker's Inequality and Lemma 6.3 in \cite{CsiszarTalata} yields  
\begin{equation}
D(\hP_{X,Y} \, ||\, P_{X,Y} ) \le \frac{c}{\kappa} \, D( P_{X,Y} \, ||\, \hP_{X,Y} ), \nn
\end{equation}
where   $\kappa:= \min_{x,y} P_{X,Y}(x,y)$ and $c=2\log 2$. Using this fact, we can use~\eqref{eqn:kl_div_bound}   to show that for all $n$ sufficiently large, 
\begin{equation}
P_{X,Y}^n(D(P_{X|Y}\, ||\, \hP_{X|Y})>\delta)\le (n+1)^{r^2} \exp(- n\delta\kappa/c),  \nn 
\end{equation}
i.e., if the arguments in the KL-divergence in~\eqref{eqn:kl_div_bound} are swapped, then the exponent is reduced by a factor proportional to $\kappa$. Using this fact and the assumption in \eqref{eqn:minentry} (uniformity of the minimum entry in the pairwise joint $\kappa>0$), we can replicate the proof of the result in~\eqref{eqn:abovelemma} with $\delta\kappa/c$  in place of $\delta$ giving
\begin{equation}
P^n(D(P\, ||\, P^*)>\delta) \le     d  (n  +  1)^{r^2}  \exp  \left(-n \delta \kappa/c \right). \nn 
\end{equation}
We then arrive at a similar result to~\eqref{eqn:final_param} by taking $\delta =(\tau \log d)/n^{1-\gamma}$. We conclude that $D(\tilP\, ||\, P^*)=O_p(d\log d/n^{1-\gamma})$. This completes the proof of the claim. 

Eq.~\eqref{eqn:risk_proj} then follows from the definition of the risk in \eqref{eqn:defrisk} and from the Pythagorean theorem  in~\eqref{eqn:Pyth2}. This implies the assertion of Theorem~\ref{cor:estimation_consist}.
\end{proof}




\bibliography{isitbib}

\end{document}